\documentclass{vldb}
\usepackage{graphicx}
\usepackage{balance}  

\usepackage{float}

\usepackage[neverdecrease]{paralist}
\usepackage{array}

\usepackage{afterpage}

\newtheorem{theorem}{Theorem}
\newtheorem{definition}{Definition}

\newtheorem{proposition}[definition]{Proposition}
\newtheorem{lemma}[definition]{Lemma}

\newtheorem{remark}[definition]{Remark}
\newtheorem{example}[definition]{Example}
\newtheorem{process}[definition]{Process}

\newcommand{\rwa}{\textsc{FrogWild}}
\newcommand{\gla}{\textsc{GraphLab PR}}

\begin{document}

% ****************** TITLE ****************************************

\title{FrogWild! -- Fast PageRank Approximations \\ on Graph Engines}

% ****************** AUTHORS **************************************

\numberofauthors{3} %  in this sample file, there are a *total*

\author{
% 1st. author
\alignauthor
Ioannis Mitliagkas\\
%\affaddr{Electrical and Computer Engineering}\\
\affaddr{ECE, UT Austin}\\
\email{ioannis@utexas.edu}
\alignauthor
Michael Borokhovich \\
%\affaddr{Electrical and Computer Engineering}\\
%\affaddr{UT Austin}\\
\affaddr{ECE, UT Austin}\\
\email{michaelbor@utexas.edu}
\and
\alignauthor
Alexandros G.~Dimakis\\
%\affaddr{Electrical and Computer Engineering}\\
%\affaddr{UT Austin}\\
\affaddr{ECE, UT Austin}\\
\email{dimakis@austin.utexas.edu}
\alignauthor
Constantine Caramanis\\
%\affaddr{Electrical and Computer Engineering}\\
%\affaddr{UT Austin}\\
\affaddr{ECE, UT Austin}\\
\email{constantine@utexas.edu}
}

\maketitle

\begin{abstract}
We propose \rwa, a novel algorithm for fast approximation of high PageRank vertices, geared towards reducing network costs of running traditional PageRank algorithms. Our algorithm can be seen as a quantized version of power iteration that performs multiple parallel random walks over a directed graph. One important innovation is that we introduce a modification to the GraphLab framework that only partially  synchronizes mirror vertices. This partial synchronization vastly reduces the network traffic generated by traditional PageRank algorithms, thus greatly reducing the per-iteration cost of PageRank. On the other hand, this partial synchronization also creates dependencies between the random walks used to estimate PageRank. Our main theoretical innovation is the analysis of the correlations introduced by this partial synchronization process and a bound establishing that our approximation is close to the true PageRank vector. 

We implement our algorithm in GraphLab and compare it against the default PageRank implementation. We show that our algorithm is very fast, performing each iteration in less than one second on the Twitter graph and can be up to $7 \times$ faster compared to the standard GraphLab PageRank implementation.
\end{abstract}

\newcommand{\dout}{d_{\mathrm{out}}}

\section{Introduction}
 
Large-scale graph processing is becoming increasingly important for the analysis of data from social networks, web pages, bioinformatics and recommendation systems. 
Graph algorithms are difficult to implement in distributed computation frameworks like Hadoop MapReduce and Spark. For this reason several in-memory graph engines like Pregel, Giraph, GraphLab and GraphX
\cite{malewicz2010pregel,graphlab,xin2013graphx,navigating} are being developed. There is no full consensus on the fundamental abstractions of graph processing frameworks but certain patterns such as vertex programming and the Bulk Synchronous Parallel (BSP) framework seem to be increasingly popular. 

PageRank computation~\cite{page1999pagerank}, which gives an estimate of the importance of each vertex in the graph, is a core component of many search routines; more generally, it represents, de facto, one of the canonical tasks performed using such graph processing frameworks. Indeed, while important in its own right, it also represents the memory, computation and communication challenges to be overcome in large scale iterative graph algorithms.

In this paper we propose a novel algorithm for fast {\em approximate calculation} of high PageRank vertices. Note that even though most previous works calculate the complete PageRank vector (of length in the millions or billions), in many graph analytics scenarios a user wants a quick estimation of the most important or relevant nodes -- distinguishing the $10^{th}$ most relevant node from the $1,000^{th}$ most relevant is important; the $1,000,000^{th}$ from the $1,001,000^{th}$ much less so. A simple solution is to run the standard PageRank algorithm for fewer iterations (or with an increased tolerance). While certainly incurring less overall cost, the per-iteration cost remains the same; more generally, the question remains whether
there is a more efficient way to approximately recover the heaviest PageRank vertices.

There are many real life applications that may benefit from a fast top-k PageRank algorithm. One example is \emph{growing loyalty of influential customers} \cite{teradata_loyal_customers}.  In this application, a telecom company identifies the top-k influential customers using the top-k PageRank on the customers' activity (e.g., calls) graph. Then, the company invests its limited budget on improving user experience for these top-k customers, since they are most important for building good reputation. 
Another interesting example is an application of PageRank for finding keywords and key sentences in a given text. In \cite{mihalcea2004textrank}, the authors show that  PageRank performs better than known machine learning techniques for keyword extraction. Each unique word (noun, verb or an adjective) is regarded as a vertex and there is an edge between two words if they occur in close proximity in the text. Using approximate top-k PageRank, we can identify the top-k keywords much faster than obtaining the full ranking. When keyword extraction is used by time sensitive applications or for an ongoing analysis of a large number of documents, speed becomes a crucial factor. 
The last example we describe here is the application of PageRank for online social networks (OSN). It is important in the context of OSNs to be able to predict which users will remain active in the network for a long time. Such \emph{key users} play a decisive role in developing effective advertising strategies and sophisticated customer loyalty programs, both vital for generating revenue \cite{heidemann:identifying}. Moreover, the remaining users can be leveraged, for instance for targeted marketing or premium services. It is shown in \cite{heidemann:identifying} that PageRank  is a much more efficient predictive measure than other centrality measures. The main innovation of \cite{heidemann:identifying} is the usage of a mixture of connectivity and activity graphs for PageRank calculation. Since these graphs are highly dynamic (especially the user activity graph), PageRank should be recalculated constantly. Moreover, the \emph{key users} constitute only a small fraction of the total number of users, thus, a fast approximation for the top-PageRank nodes constitutes a desirable alternative to the exact solution.

In this paper we address this problem. Our algorithm (called \rwa\ for reasons that will become subsequently apparent) significantly outperforms the simple reduced iterations heuristic in terms of {\em running time, network communication and scalability}. We note that, naturally, we compare our algorithm and reduced-iteration-PageRank within the same framework: we implemented our algorithm in GraphLab PowerGraph and compare it against the built-in PageRank implementation. A key part of our contribution also involves the proposal of what appears to be simply a technically minor modification within the GraphLab framework, but nevertheless results in significant network-traffic savings, and we believe may nevertheless be of more general interest beyond PageRank computations. 

\textbf{Contributions:} We consider the problem of fast and efficient (in the sense of time, computation and communication costs) computation of the high PageRank nodes, using a graph engine. To accomplish this we propose and analyze an new PageRank algorithm specifically designed for the graph engine framework, and, significantly, we propose a modification of the standard primitives of the graph engine framework (specifically, GraphLab PowerGraph), that enables significant network savings. We explain in further detail both our objectives, and our key innovations. 

Rather than seek to recover the full PageRank vector, we aim for the top $k$ PageRank vertices (where $k$ is considered to be approximately in the order of $10-1000$). Given an output of a list of $k$ vertices, we define two natural accuracy metrics
that compare the true top-$k$ list with our output.  The algorithm we propose, \rwa\,  operates by starting a small (sublinear in the number of vertices $n$) number of random walkers (\textit{frogs}) that jump randomly on the directed graph. The random walk interpretation of PageRank enables the frogs to jump to a completely random vertex (teleport) with some constant probability (set to $0.15$ in our experiments, following standard convention). After we allow the frogs to jump for time equal to the mixing time of this non-reversible Markov chain, their positions are sampled from the invariant distribution $\pi$ which is normalized PageRank. The standard PageRank iteration can be seen as the continuous limit of this process (\textit{i.e.}, the frogs become water), which is equivalent to power iteration for stochastic matrices. 

The main {\em algorithmic contributions} of this paper are comprised of the following three innovations. First, we argue that discrete frogs (a quantized form of power iteration) is significantly better for distributed computation when one is interested only in the large entries of the eigenvector $\pi$. This is because each frog produces an independent sample from $\pi$. If some entries of $\pi$ are substantially larger and we only want to determine those, a small number of independent samples suffices. We make this formal using standard Chernoff bounds (see also \cite{sarma2011estimating,das2013distributed} for similar arguments). On the contrary, during standard PageRank iterations, vertices pass messages to all their out-neighbors since a non-zero amount of water must be transferred. This tremendously increases the network bandwidth especially when the graph engine is over a cluster with many machines. 

One major issue with simulating discrete frogs in a graph engine is teleportations. Graph frameworks partition vertices to physical nodes and restrict communication on the edges of the underlying graph. Global random jumps would create dense messaging patterns that would increase communication. Our second innovation is a way of obtaining an identical sampling behavior without teleportations. We achieve this by initiating the frogs at uniformly random positions and having them perform random walks for a life span that follows a geometric random variable. The geometric probability distribution depends on the teleportation probability and can be calculated explicitly. 

Our third innovation involves a simple proposed modification for graph frameworks. Most modern graph engines (like GraphLab PowerGraph~\cite{gonzalez2012powergraph}) employ vertex-cuts as opposed to edge-cuts. This means that each vertex of the graph is assigned to multiple machines so that graph edges see a local vertex mirror. One copy is assigned to be the master and maintains the master version of vertex data while remaining replicas are mirrors that maintain local cached read--only copies of the data. Changes to the vertex data are made to the master and then replicated to all mirrors at the next synchronization barrier. This architecture is highly suitable for graphs with high-degree vertices (as most real-world graphs are) but has one limitation when used for a few random walks: imagine that vertex $v_1$ contains one frog that wants to jump to $v_2$. 
If vertex $v_1$ has very high degree, it is very likely that multiple replicas of that vertex exist, possibly one in each machine in the cluster. In an edge-cut scenario only one message would travel from $v_1 \rightarrow v_2$, assuming $v_1$ and $v_2$ are located in different physical nodes. However, when vertex-cuts are used, the state of $v_1$ is updated (i.e., contains no frogs now) and this needs to be communicated to all mirrors. It is therefore possible that \textit{a single random walk can create a number of messages equal to the number of machines in the cluster}. 

We modify PowerGraph to expose a scalar parameter $p_s$ per vertex. By default, when the framework is running, in each super-step all masters synchronize their programs and vertex data with their mirrors. Our modification is that for each mirror we flip an independent coin and synchronize with probability $p_s$. Note that when the master does not synchronize the vertex program with a replica, that replica will not be active during that super-step. Therefore, we can avoid the communication and CPU execution by performing limited synchronization in a randomized way. 

\rwa\ is therefore executed asynchronously but relies on the Bulk Synchronous execution mode of PowerGraph with the additional simple randomization we explained. The name of our algorithm is inspired by HogWild~\cite{hogwild}, a lock-free asynchronous stochastic gradient descent algorithm proposed by Niu \textit{et al.}. We note that PowerGraph does support an asynchronous execution mode~\cite{gonzalez2012powergraph} but we implemented our algorithm by a small modification of synchronous execution. As discussed in~\cite{gonzalez2012powergraph}, the design of asynchronous graph algorithms is highly nontrivial and involves locking protocols and other complications. Our suggestion is that for the specific problem of simulating multiple random walks on a graph, simply randomizing synchronization can give significant benefits while keeping design simple. 

While the parameter $p_s$ clearly has the power to significantly reduce network traffic -- and indeed, this is precisely born out by our empirical results -- it comes at a cost: the standard analysis of the Power Method iteration no longer applies. The main challenge that arises is the theoretical analysis of the \rwa\  algorithm. The model is that each vertex is separated across machines and each connection between two vertex copies is present with probability $p_s$. A single frog performing a random walk on this new graph defines a new Markov Chain and this can be easily designed to have the same invariant distribution $\pi$ equal to normalized PageRank. The complication is that the trajectories of frogs are {\em no longer independent}: if two frogs are in vertex $v_1$ and (say) only one mirror $v'_1$ synchronizes, both frogs will need to jump through edges connected with that particular mirror. Worse still, this correlation effect increases, the more we seek to improve network traffic by further decreasing $p_s$. Therefore, it is no longer true that one obtains independent samples from the invariant distribution $\pi$. Our theoretical contribution is the development of an analytical bound that shows that these dependent random walks still can be used to obtain $\hat{\pi}$ that is provably close to $\pi$ with high probability. We rely on a coupling argument combined with an analysis of pairwise intersection probabilities for random walks on graphs. In our convergence analysis we use the {\em contrast bound} \cite{pierre1999markov} for non-reversible chains.

\newcommand{\simplex}{\Delta}
\newcommand{\simplexn}{\simplex^{n-1}}
\newcommand{\pmeet}{p_{\cap}}

\subsection{Notation}
Lowercase letters denote scalars or vectors.
Uppercase letters denote matrices.
The $(i,j)$ element of a matrix $A$ is $A_{ij}$.
We denote 
the transpose of a matrix $A$ by $A'$.
For a time-varying vector $x$, we denote its value at time $t$ by $x^t$. 
When not otherwise specified, $\|x\|$ denotes the $l_2$-norm of vector x.
We use $\simplexn$ for the probability simplex in $n$ dimensions, and
and $e_i \in \simplexn$ for the indicator vector for item $i$. For example,
$e_1 = [1, 0, ... 0]$.
For the set of all integers from $1$ to $n$ we write $[n]$.

\section{Problem and Main Results}
We now make precise the intuition and outline given in the introduction. We first define the problem, giving the definition of PageRank, the PageRank vector, and therefore its top elements. We then define the algorithm, and finally state our main analytical results.

\subsection{Problem Formulation}
\label{sec:problemformulation}
Consider a directed graph $G=(V,E)$ with $n$ vertices ($|V| = n$) and let $A$ denote its adjacency matrix. That is, $A_{ij}=1$ if there is an edge from $j$ to $i$. Otherwise, the value is $0$. Let $\dout(j)$ denote the number of {\em successors} (out-degree) of vertex $j$ in the graph.
We assume that all nodes have at least one successor, $\dout(j)>0$.
Then we can define the transition probability matrix $P$ as follows:
\begin{equation}
\label{eqn:tpm}
	P_{ij} = A_{ij} / \dout(j).
\end{equation}
The matrix is left-stochastic, which means that each of its rows sums to $1$. We call $G(V,E)$ the {\it original graph}, as opposed to the PageRank graph which includes a probability of transitioning to any given vertex. We now define this transition probability matrix, and the PageRank vector.

\begin{definition}[PageRank \cite{page1999pagerank}]
\label{def:pagerank}
Consider the matrix 
\[
		Q \triangleq (1-p_T)  P + p_T {1 \over n} 1_{n \times n}.
\]
where $p_T \in [0, 1]$ is a parameter, most commonly set to $0.15$.
The PageRank vector $\pi \in \simplexn$ is defined as the principal right eigenvector of $Q$. That is,
$
		\pi \triangleq v_1(Q).
$
By the Perron-Frobenius theorem, the corresponding eigenvalue is $1$. This implies the {\em fixed-point} characterization of the PageRank vector,
$
		\pi = Q \pi
$.
\end{definition}

The PageRank vector assigns high values to {\em important} nodes. Intuitively, important nodes have many important predecessors (other nodes that point to them).
 This recursive definition is what makes PageRank robust to manipulation, but also expensive to compute. It can be recovered by exact eigendecomposition of $Q$, but at real problem scales this is prohibitively expensive. In practice, engineers often use a few iterations of the power method to get a "good-enough" approximation.

The definition of PageRank hinges on the left-stochastic matrix $Q$, suggesting a connection to Markov chains. Indeed, this connection is well documented and studied \cite{agarwal2007learning,fortunato2007random}. An important property of PageRank from its random walk characterization, is the fact that $\pi$ is the invariant distribution for a Markov chain with dynamics described by $Q$. A non-zero $p_T$, also called the {\em teleportation probability}, introduces a uniform component to the PageRank vector $\pi$. We see in our analysis that this implies ergodicity and faster mixing for the random walk.

\subsubsection{Top PageRank Elements}

Given the true PageRank vector, $\pi$ and an estimate $v$ given by an approximate PageRank algorithm, we define the top-$k$ accuracy using two metrics. 

\begin{definition}[Mass Captured] 
\label{def:capturedmass}
Given a distribution $v \in \simplexn$, the true PageRank distribution $\pi \in \simplexn$ and an integer $k \geq 0$, we define the {\em mass captured} by $v$ as follows.
\[
	\mu_k(v) \triangleq  \pi(\mathrm{argmax}_{|S|=k}v(S))
\]
For a set $S \in [n]$, $v(S) = \sum_{i \in S} v(i)$ denotes the total mass ascribed to the set by the distribution $v \in \simplexn$.
\end{definition}

Put simply, the set $S^*$ that gets the most mass according to $v$ out of all sets of size $k$, is evaluated according to $\pi$ and that gives us our metric. It is maximized by $\pi$ itself, i.e. the optimal value is $\mu_k(\pi)$.

The second metric we use is the exact identification probability, \textit{i.e.} the fraction of elements in the output list that are also in the true top-$k$ list. Note that the second metric is limited in that it does not give partial credit 
for high PageRank vertices that were not in the top-$k$ list. In our experiments in Section~\ref{sec:experiments}, we mostly use the normalized captured mass accuracy metric but also report the exact identification probability for some cases -- typically the results are similar. 

We subsequently describe our algorithm. We attempt to approximate the heaviest elements of the invariant distribution of a Markov Chain, by simultaneously performing multiple random walks on the graph. The main modification to PowerGraph, is the exposure of a parameter, $p_s$, that controls the probability that a given master node synchronizes with any one of its mirrors. Per step, this leads to a proportional reduction in network traffic. The main contribution of this paper is to show that we get results of comparable or improved accuracy, while maintaining this network traffic advantage. We demonstrate this empirically in Section~\ref{sec:experiments}.

\vspace{2mm}
\subsection{Algorithm}
\label{sec:algo}

During setup, the graph is partitioned using GraphLab's default ingress algorithm. At this point each one of $N$ frogs
is born on a vertex chosen uniformly at random. Each vertex $i$ carries a counter initially set to $0$ and denoted by $c(i)$. Scheduled vertices execute the following program.

Incoming frogs from previously executed vertex programs, are collected by the \texttt{init()} function. At \texttt{apply()} every frog dies with probability 
$p_T = 0.15$. This, along with a uniform starting position, effectively simulates the $15\%$ uniform component from Definition~\ref{def:pagerank}.

A crucial part of our algorithm is the change in  synchronization behaviour. The $<$\texttt{sync}$>$ step only synchronizes a $p_s$ fraction of mirrors leading to commensurate gains in network traffic (cf.\ Section~\ref{sec:experiments}). This patch on the GraphLab codebase was only a few lines of code. Section~\ref{sec:experiments} contains more details regarding the implementation.

The \texttt{scatter()} phase is only executed for edges $e$ incident to a mirror of $i$ that has been synchronized. Those edges draw a binomial number of frogs to send to their other endpoint.
The rest of the edges perform no computation.
The frogs sent to vertex $j$ at the last step will be collected at the \texttt{init()} step when $j$ executes.

\vspace{0.1in}
\noindent\begin{minipage}{.48\textwidth}
\rule{\textwidth}{1pt}

\vspace{0.05in}

{\large \textbf{FrogWild!} }vertex program

\vspace{-0.07in}
\rule{\textwidth}{1pt}

{\small

\textbf{Input parameters}: $p_s$, $p_T=0.15$, $t$
\begin{itemize}
\item[\textbf{apply($i$)}] $K(i) \leftarrow [\textrm{\# incoming frogs}]$
\item[] If $t$ steps have been performed, $c(i) \leftarrow c(i)+K(i)$ and halt.
\item[] For every incoming frog: 
\begin{itemize}
\item[] With probability $p_T$, frog dies:
\begin{itemize}
\item[] $c(i) \leftarrow c(i)+1$,
\item[] $K(i) \leftarrow K(i)-1$.
\end{itemize}
\end{itemize}
\item[$<$\textbf{sync}$>$]  For every {\em mirror} $m$ of vertex $i$:
\begin{itemize}
\item[] With probability $p_s$:
\begin{itemize}
\item[] Synchronize state with mirror $m$.
\end{itemize}
\end{itemize}
\item[\textbf{scatter($e=(i,j)$)}] [Only on synchronized mirrors]

\item[] Generate Binomial number of frogs: \[x \sim Bin\left(K(i), { 1 \over \dout(i) p_s}\right) \]
\item[] Send $x$ frogs to vertex $j$: \texttt{signal(j,x)}
\end{itemize}
\vspace{0.05in}

}
\vspace{-0.15in}
\rule{\textwidth}{1pt}

\end{minipage} 
\vspace{0.1in}

Parameter $p_T$ is the teleportation probability from the random surfer model in \cite{page1999pagerank}. To get PageRank using random walks, one could adjust the transition matrix $P$ as described in Definition~\ref{def:pagerank} to get the matrix $Q$. Alternatively, the process can be replicated by a random walk following the original matrix $P$, and teleporting at every time, with probability $p_T$. The destination for this teleportation is chosen uniformly at random from $[n]$. We are interested in the position of a walk at a predetermined point in time as that would give us a sample from $\pi$. This holds as long as we allow enough time for mixing to occur.

Due to the inherent markovianity in this process, one could just consider it starting from the last teleportation before the predetermined stopping time. When the mixing time is large enough, the number of steps performed between the last teleportation and the predetermined stopping time, denoted by $X$, is geometrically distributed with parameter $p_T$. This follows from the time-reversibility in the teleportation process: inter-teleportation times are geometrically distributed, so as long as the first teleportation event happens before the stopping time, then $X \sim \mathrm{Geom(p_T)}$.

This establishes that, the \rwa !  process -- where a frog performs a geometrically distributed number of steps following the original transition matrix $P$ -- closely mimics a random walk that follows the adjusted transition matrix, $Q$. In practice, we stop the process after $t$ steps to get a good approximation. To show our main result, Theorem~\ref{thm:main},  we analyze the latter process.

Using a binomial distribution to independently generate the number of frogs in the \texttt{scatter()} phase closely models the effect of random walks. The marginal distributions are correct, and the number of frogs, that did not die during the \texttt{apply()} step, is preserved in expectation. For our implementation we resort to a more efficient approach. Assuming $K(i)$ frogs survived the \texttt{apply()} step, and $M$ mirrors where picked for synchronization, then we send $\lceil {K(i) \over M }\rceil$ frogs to $\min(K(i), M)$ mirrors. If the number of available frogs is less than the number of synchronized mirrors, we pick $K(i)$ arbitrarily.

\subsection{Main Result}

Our analytical results essentially provide a high probability guarantee that our algorithm produces a solution that approximates well the PageRank vector. Recall that the main modification of our algorithm involves randomizing the synchronization between master nodes and mirrors. For our analysis, we introduce a broad model to deal with partial synchronization, in Appendix~\ref{sec:erasuremodel}.

Our results tell us that partial synchronization does not change the distribution of a single random walk. To make this and our other results clear, we need the simple definition.
\begin{definition}
We denote the state of random walk $i$ at its $t^{th}$ step  by $s_i^t$.
\end{definition}

Then, we see that $\mathbb{P}\big( s_1^{t+1} = i \big| s_1^t = j  \big) = { 1 /\dout(j)}$,
and $x_1^{t+1} = P x_1^t.$
This follows simply by the symmetry assumed in Definition~\ref{def:erasuremodel}.
Thus if we were to sample in serial, the modification of the algorithm controlling (limiting) synchronization would not affect each sample, and hence would not affect our estimate of the invariant distribution. However, we start multiple (all) random walks simultaneously. In this setting, the fundamental analytical challenge stems from the fact that any set of random walks with intersection are now correlated. The key to our result is that we can control the effect of this correlation, as a function the parameter $p_s$ and the {\it pairwise probability} that two random walks intersect. We define this formally.
\begin{definition}
\label{def:pmeet}
Suppose two walkers $l_1$ and $l_2$ start at the same time and perform $t$ steps. The probability that they meet is defined as follows.
\begin{equation}
	\pmeet(t) \triangleq 
	\mathbb{P}\left(
			\exists\ \tau \in [0, t],\ \mathrm{s.t.}\ 
			s_{l_1}^{\tau} = s_{l_2}^{\tau} 
	\right)
\end{equation}
\end{definition}

\begin{definition}[Estimator]
\label{def:estimator}
Given the positions of $N$ random walks at time $t$, $\{s_l^t\}_{l=1}^N$, we define the following estimator for the invariant distribution $\pi$.
\begin{equation}
	\hat{\pi}_N(i) \triangleq { \left| \{  l  : l \in [N], s_l^{t}=i   \}\right| \over N  } = { c(i) \over N} 
\end{equation}
Here $c(i)$ refers to the tally maintained by the \rwa! vertex program.
\end{definition}

Now we can state the main result. Here we give a guarantee for the quality of the solution furnished by our algorithm. 

\begin{theorem}[Main Theorem]
\label{thm:main}
Consider $N$ frogs following the \rwa!\ process (Section~\ref{sec:algo}), under the erasure model of Definition~\ref{def:erasuremodel}. The frogs start at independent locations, distributed uniformly and stop after a geometric number of steps or, at most, $t$ steps.
The estimator $\hat{\pi}_N$ (Definition~\ref{def:estimator}),
captures mass close to the optimal. Specifically, with probability at least $1-\delta$,
\[
	\mu_k(\hat{\pi}_N) \geq \mu_k(\pi) - \epsilon,
\]
where 
\begin{equation}
\label{eqn:mainresult}
\epsilon <  \sqrt{(1-p_T)^{t+1} \over p_T }
   +  \sqrt{
		{ k \over \delta}
		\left[
			{ 1  \over N } + (1 - p_s^2)\pmeet(t)
		\right]
	}.
\end{equation}
\end{theorem}

\begin{remark}[Scaling]
\label{thm:scaling}
The result in Theorem~\ref{thm:main} immediately implies the following scaling for the number of iterations and frogs respectively. They both depend on the maximum captured mass possible, $\mu_k(\pi)$ and are sufficient for making the error, $\epsilon$, of the same order as $\mu_k(\pi)$.
\[
	t = O\left(\log{1 \over \mu_k(\pi)}\right), \qquad
	N = O\left({k \over \mu_k(\pi)^2}\right)
\]
\end{remark}

The proof of Theorem~\ref{thm:main} is deferred to Appendix~\ref{sec:proofmain}.
The guaranteed accuracy via this result also depends on the probability that two walkers will intersect. Via a simple argument, that probability is the same as the meeting probability for independent walks. The next theorem calculates this probability.

\begin{theorem}[Intersection Probability]
\label{thm:pmeet}

Consider \\
two independent random walks obeying the same ergodic transition probability matrix, $Q$ with invariant distribution $\pi$, as described in Definition~\ref{def:pagerank}. Furthermore, assume that both of them are initially distributed uniformly over the state space of size $n$. The probability that they meet within $t$ steps, is bounded as follows,
\[
	\pmeet(t) \leq {1 \over n } + {t \| \pi \|_\infty \over p_T},
\]
where $\| \pi \|_\infty$, denotes the maximal element of the vector $\pi$.
\end{theorem}

The proof is based on the observation that the $l_\infty$ norm of a distribution controls the probability that two independent samples coincide. We show that for all steps of the random walk, that norm is controlled by the $l_\infty$ norm of $\pi$. We defer the full proof to Appendix~\ref{sec:proofpmeet}.

A number of studies, give experimental evidence (e.g.\ \cite{becchetti2006distribution}) suggesting that PageRank values for the web graph follow a power-law distribution with parameter approximately $\theta=2.2$.  That is true for the tail of the distribution -- the largest values, hence of interest to us here -- regardless of the choice of $p_T$. The following proposition bounds the value of the heaviest PageRank value, $\|\pi\|_\infty$.

\begin{proposition}[Max of Power-Law Distribution]
\label{thm:powerlawmax}
Let $\pi \in \simplexn$ follow a power-law distribution with parameter $\theta$ and minimum value $p_T/n$.
Its maximum element, $\|\pi\|_\infty$, is at most $n^{-\gamma}$, with probability at least $1-cn^{\gamma-{1\over \theta-1}}$, for some universal constant $c$.
\end{proposition}
Assuming $\theta=2.2$ and picking, for example, $\gamma=0.5$, we get that 
\[
	\mathbb{P}(\| \pi \|_\infty > 1 / \sqrt{n}) \leq c n^{-1 / 3}.
\]
This implies that with probability at least $1-cn^{-{1/3}}$ the meeting probability is bounded as follows.
\[
	\pmeet(t) \leq {1 \over n } + {t  \over  p_T \sqrt{n}}.
\]
One would usually take a number of steps $t$ that are either constant or logarithmic with respect to the graph size $n$.
This implies that for many reasonable choices of set size $k$ and acceptable probability of failure $\delta$, the meeting probability vanishes as $n$ grows. Then we can make the 
 second term of the error in \eqref{eqn:mainresult} arbitrarily small by controlling the number of frogs, $N$.
The proof for Proposition~\ref{thm:powerlawmax} is deferred to Appendix~\ref{sec:proofpowerlawmax}.

\subsection{Prior Work}
\label{sec:prior_work}

There is a very large body of work on computing and approximating PageRank on different computation models (\textit{e.g.} see~\cite{berkhin2005survey,broder2006efficient,sarma2011estimating,das2013distributed,andersen2007local} and references therein). To the best of our knowledge, our work is the first to specifically design an approximation algorithm for high-PageRank nodes for graph engines.
Another line of work looks for {\em Personalized PageRank} (PPR) scores. This quantifies the influence an arbitrary node $i$ has on another node $j$, cf. recent work \cite{lofgren2014fast} and discussion therein. In \cite{avrachenkov_top_k_personal_2010}, the top-k approximation of PPR is studied. However, PPR is not applicable in our case, as we are looking for an answer close to a {\em global} optimum.

In \cite{avrachenkov_top_k_global_2007}, a random-walks-based algorithm  is proposed. The authors provide some insightful analysis of different variations of the algorithm. They show that starting a single walker from every node, is sufficient to achieve a good global approximation. We focus on capturing a few nodes with a lot of mass, hence we can get away with orderwise much fewer frogs than $O(n)$.  This is important for achieving low network traffic when the algorithm is executed on a distributed graph framework. Figure \ref{fig:plots532} shows linear reduction in network traffic when the number of initial walkers decreases. Furthermore, our method does not require waiting for the last frog to naturally expire (note that the geometric distribution has infinite support). We impose a very short time cut-off, $t$, and exactly analyze the  penalty in captured mass we pay for it in Theorem~\ref{thm:main}.

One natural question is how our algorithm compares to, or can be complemented by, graph sparsification techniques. One issue here is that graph sparsification crucially depends on the similarity metric used. Well-studied properties that are preserved by different sparsification methods involve lengths of shortest paths between vertices (such sparsifiers are called Spanners, see e.g. \cite{Peleg:1987}), cuts between subsets of vertices \cite{Benczur:1996} and more generally quadratic forms of the graph laplacian \cite{Spielman:2011,Batson:2013}, see \cite{Batson:2013} and references therein for a recent overview. To the best of our knowledge, there are no known graph sparsification techniques that preserve vertex PageRank.

One natural heuristic that one may consider is to independently flip a coin and delete each edge of the graph with some probability $r$. Note that this is crucially different from spectral sparsifiers \cite{Spielman:2011,Batson:2013} that choose these probabilities using a process that is already more complicated than estimating PageRank.
This simple heuristic of independently deleting edges indeed accelerates the estimation process for high-PageRank vertices. We compare $\rwa$ to this uniform sparsification process in Figure \ref{fig:plots530}.  We present here results for 2 iterations of the $\gla$ on the sparsified graph.
Note that running only one iteration is not interesting since it actually estimates only the in-degree of a node which is known in advance (i.e., just after the graph loading) in a graph engine framework. It can be seen in Figure \ref{fig:plots530} that even when only two iterations are used on the sparsified graph the running time is significantly worse compared to $\rwa$ and the accuracy is comparable.

Our base-line comparisons come from the graph framework papers since PageRank is a standard benchmark for running-time, network and other computations. Our implementation is on GraphLab (PowerGraph) 
and significantly outperforms the built-in PageRank algorithm. This algorithm is already shown in~\cite{gonzalez2012powergraph,navigating} to be significantly more efficient compared to other frameworks like Hadoop, Spark, Giraph \textit{etc.}

\section{Experiments}
\label{sec:experiments}

\begin{figure*}[!ht]
\centering
\includegraphics[width=1\columnwidth]{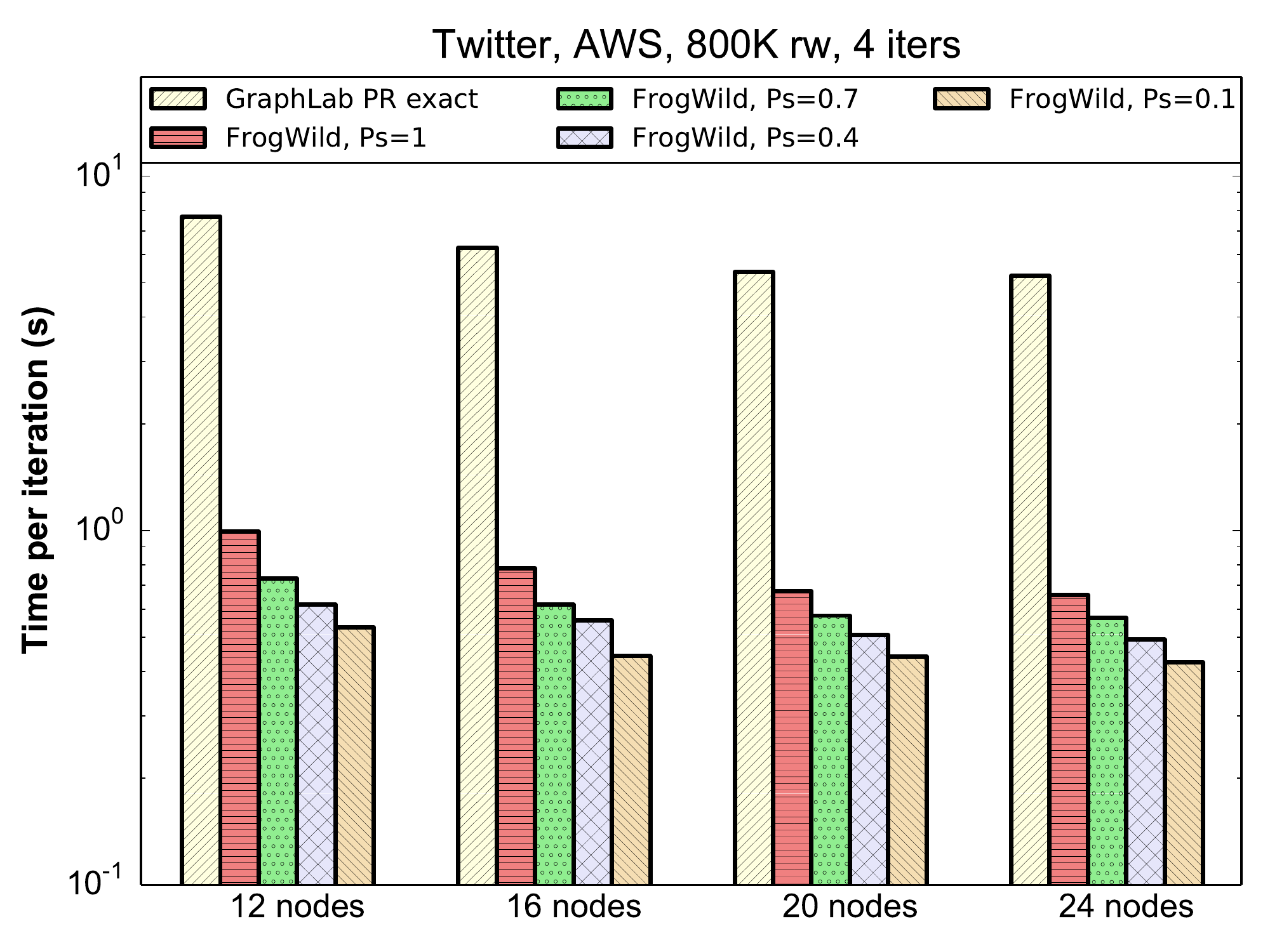}
\includegraphics[width=1\columnwidth]{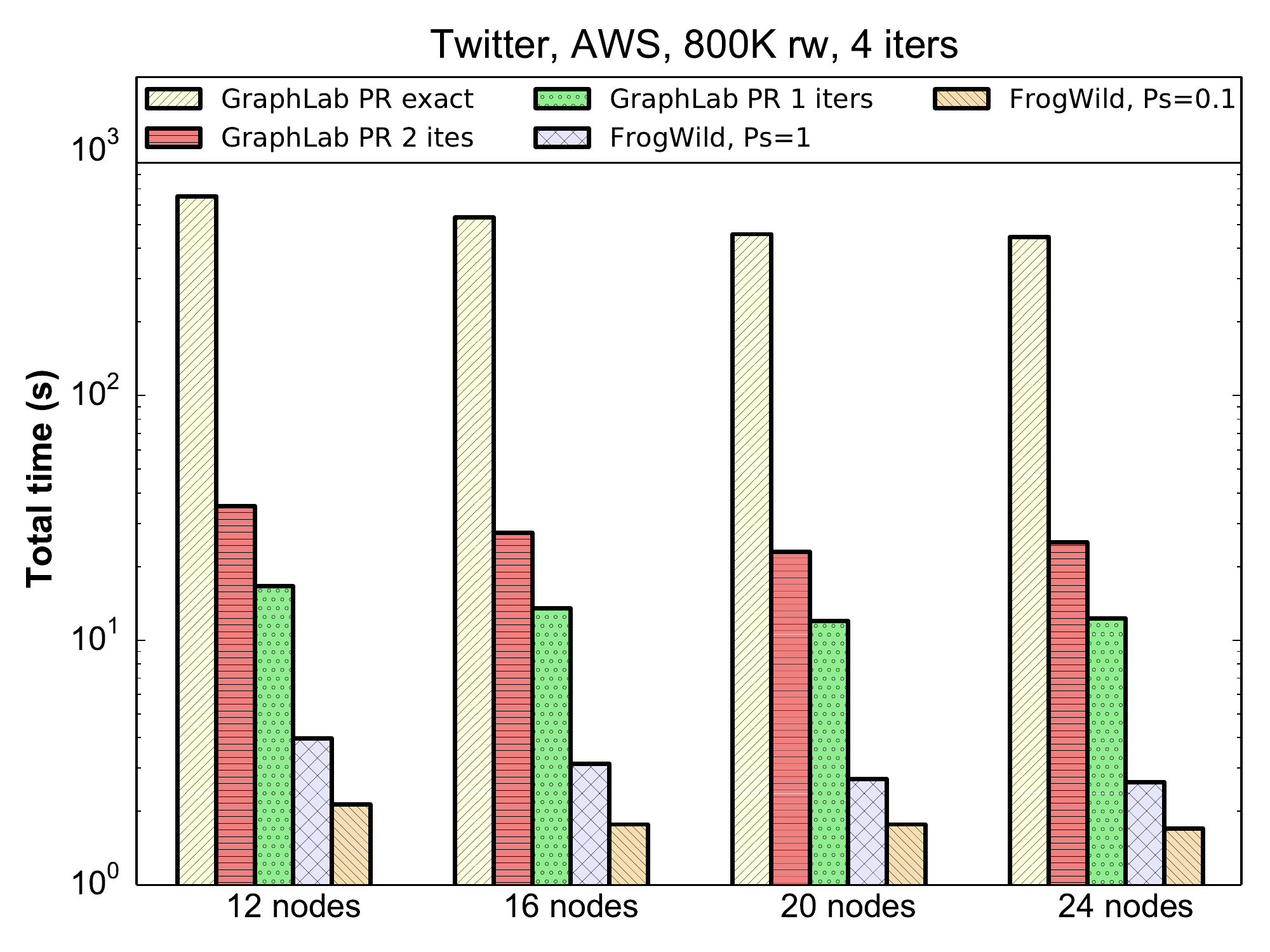}\\
{\small(a)\hspace{\columnwidth}(b)}\\\vspace{2mm}
\includegraphics[width=1\columnwidth]{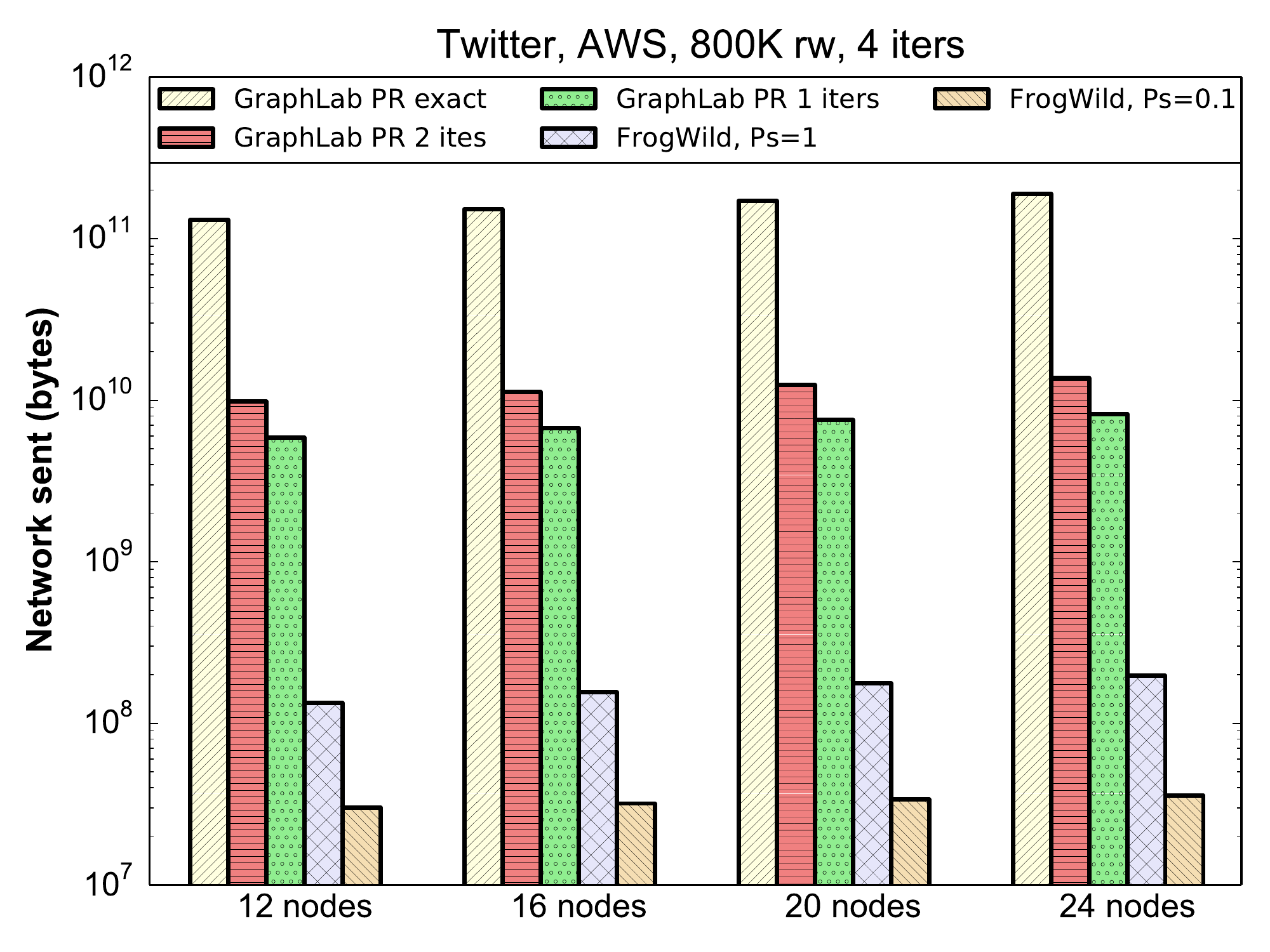}
\includegraphics[width=1\columnwidth]{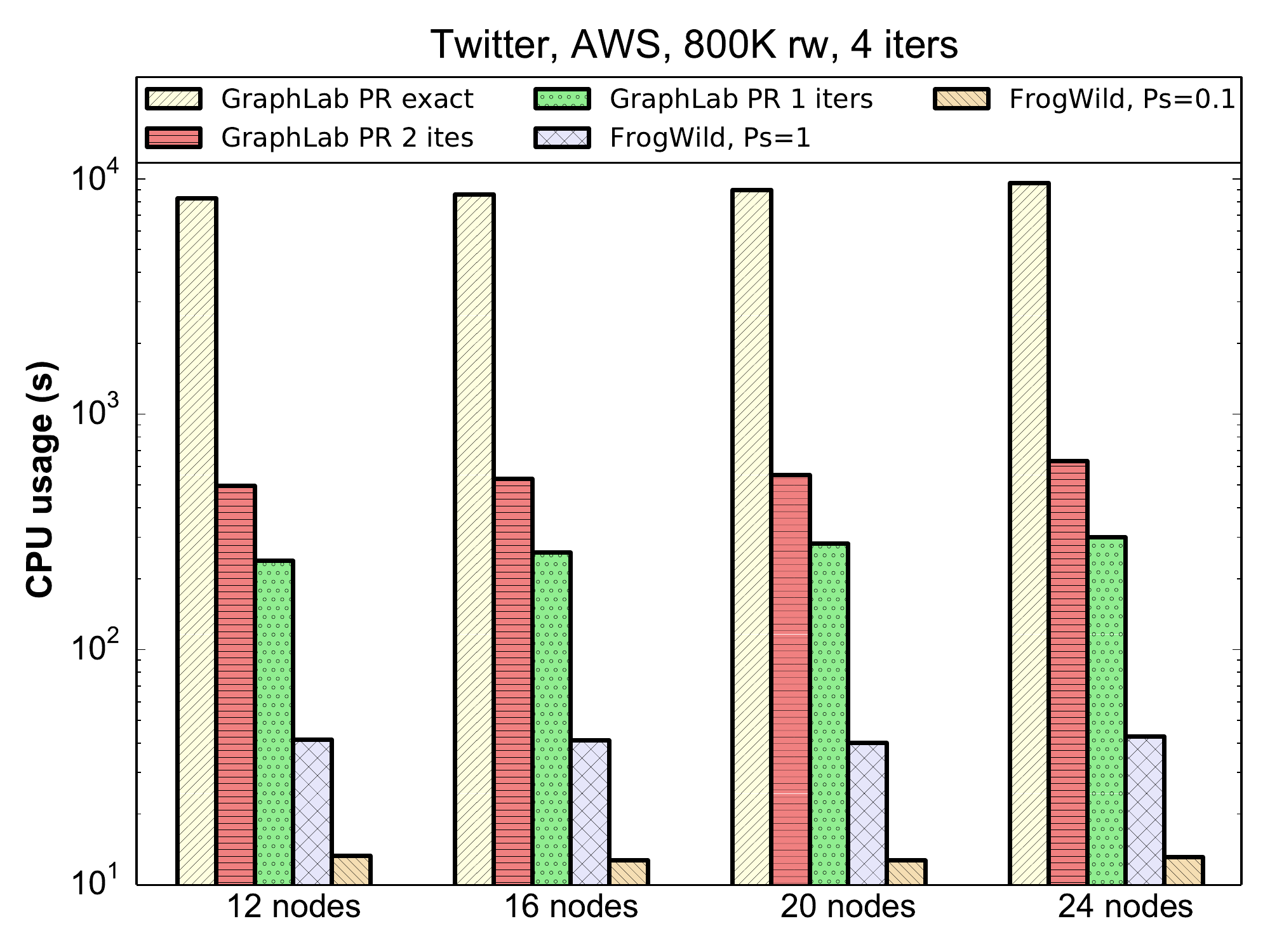}\\
{\small(c)\hspace{\columnwidth}(d)}
\caption{\small PageRank performance for various number of nodes. Graph: Twitter; system: AWS (Amazon Web Services); $\rwa$ parameters: 800K initial random walks and 4 iterations. 
(a) -- Running time per iteration.
(b) -- Total running time of the algorithms.
(c) -- Total network bytes sent by the algorithm during the execution (does not include ingress time).
(d) -- Total CPU usage time. Notice, this metric may be larger than the total running time since many CPUs run in parallel. 
}\label{fig:plots221}
\end{figure*}

\begin{figure*}[!ht]
\centering
\includegraphics[width=\columnwidth]{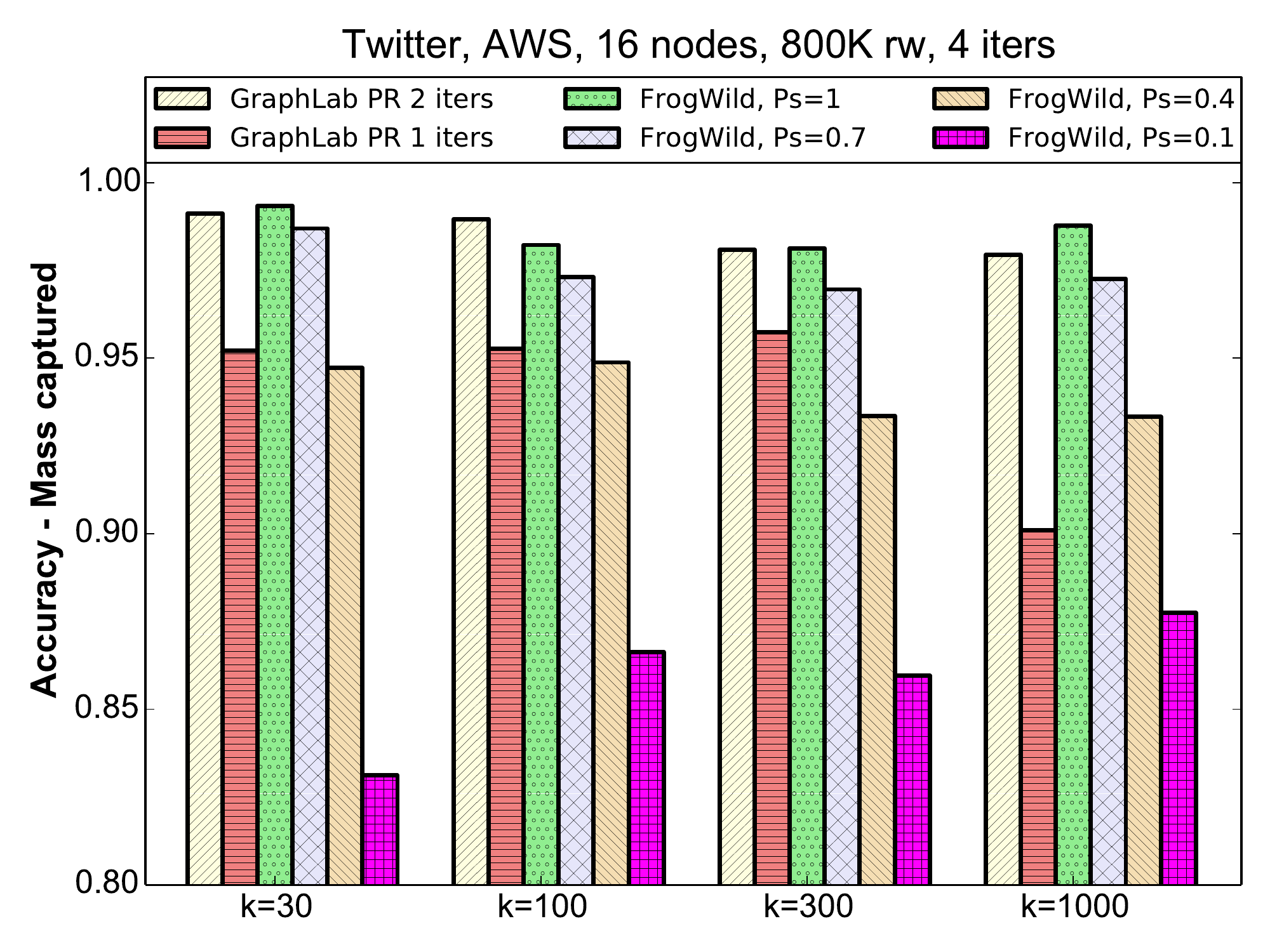}
\includegraphics[width=\columnwidth]{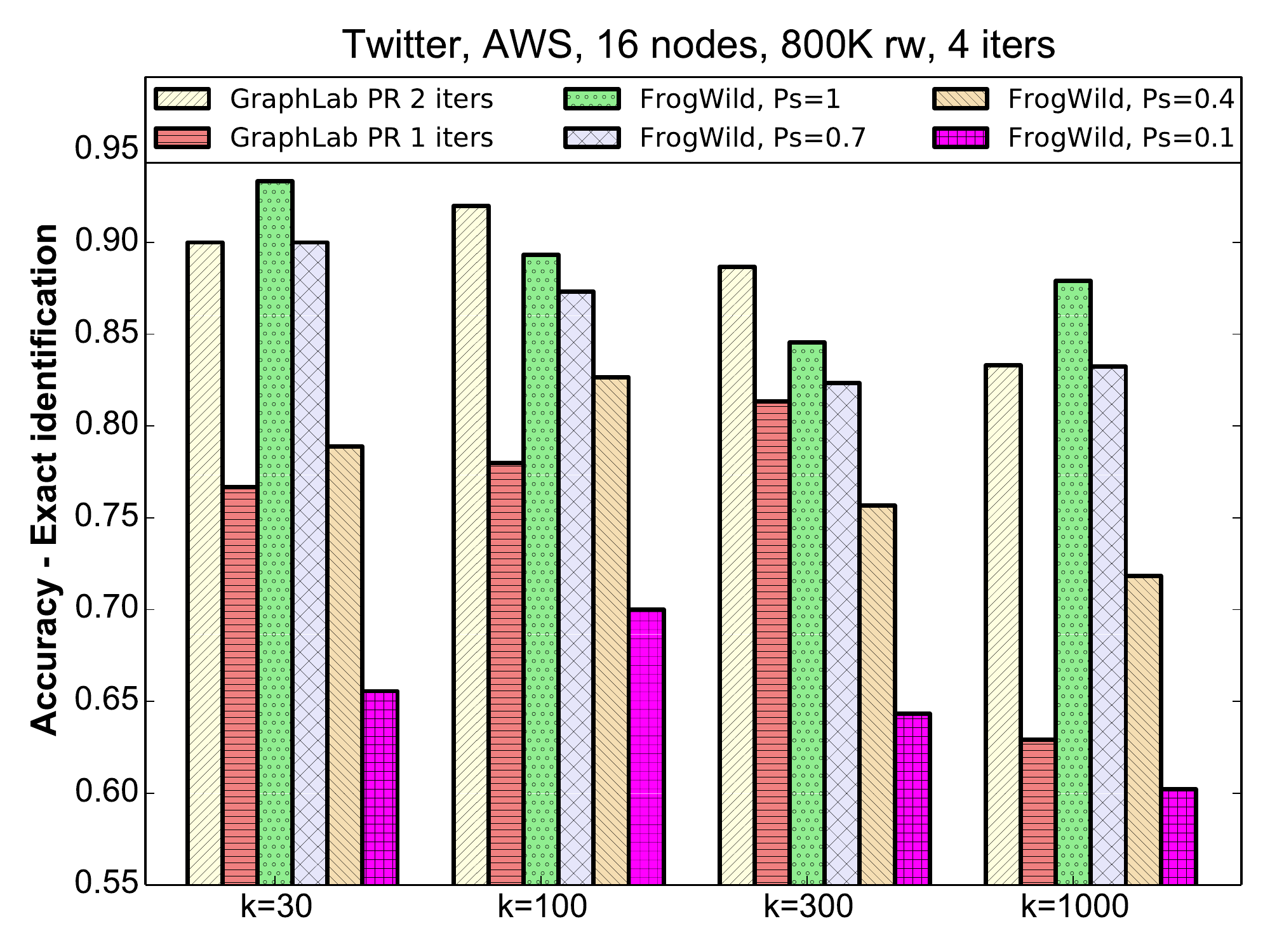}\\
{\small(a)\hspace{\columnwidth}(b)}\\
\caption{\small 
PageRank approximation accuracy for various number of top-$k$ PageRank vertices. Graph: Twitter; system: AWS (Amazon Web Services) with 16 nodes; $\rwa$ parameters: 800K initial random walks and 4 iterations.
(a) -- Mass captured. The total PageRank that the reported top-$k$ vertices worth in the exact ranking.
(b) -- Exact identification. The number of vertices in the intersection of the reported top-$k$ and the exact top-$k$ lists.
}\label{fig:plots222}
\end{figure*}

In this section we compare the performance of our algorithm to the PageRank algorithm shipped with GraphLab v2.2 (PowerGraph) \cite{graphlab}. The fact that GraphLab is the fastest distributed engine for PageRank is established experimentally in \cite{navigating}.  We focus on two algorithms: the basic built-in algorithm provided as part of the GraphLab \emph{graph analytics toolkit}, referred to here as $\gla$, and $\rwa$.
Since we are looking for a top-$k$ approximation and $\gla$ is meant to find the entire PageRank vector, we only run it for a small number of iterations (usually $2$ are sufficient). This gives us a good top-$k$ approximation and is much faster than running the algorithm until convergence. We also fine tune the algorithm's tolerance parameter to get a good but fast approximation.

We compare several performance metrics, namely: running time, network usage, and accuracy. The metrics do not include time and network usage required for loading the graph into  GraphLab (known as the \emph{ingress time}). They reflect only the execution stage.
 
\subsection{The Systems}
 
We perform experiments on two systems. 
The first system is a cluster of $20$ virtual machines, created using VirtualBox 4.3 \cite{virtbox} on a single physical server. The server is based on an Intel\textsuperscript{\textregistered} Xeon\textsuperscript{\textregistered} CPU E5-1620 with 4 cores at 3.6 GHz, and 16 GB of RAM. 
The second system, comprises of a cluster of up to 24 EC2 machines on AWS (Amazon web services) \cite{aws}. We use m3.xlarge instances, based on Intel\textsuperscript{\textregistered} Xeon\textsuperscript{\textregistered} CPU E5-2670 with  4 vCPUs and 15 GB RAM.

\subsection{The Data}
 
For the VirtualBox system, we use the LiveJournal graph \cite{snapnets} with 4.8M vertices and 69M edges. For the AWS system, in addition to the LiveJournal graph, we use the Twitter graph \cite{Kwak10www} which has 41.6M nodes and 1.4B edges.

\subsection{Implementation}
$\rwa$ is implemented on the standard GAS (gather, apply, scatter) model. We implement  \texttt{init()}, \texttt{apply()}, and \texttt{scatter()}. The purpose of  \texttt{init()} is to collect the random walks sent to the node by its neighbors using \texttt{scatter()} in the previous iteration. In the first iteration, \texttt{init()} generates a random fraction of the initial total number of walkers. This implies that the initial walker locations are randomly distributed across nodes.
$\rwa$ requires the length of random walks to be geometrically distributed (see Section \ref{sec:algo}). For the sake of efficiency, we impose an upper bound on the length of random walks. The algorithm is executed for the constant number of iterations (experiments show good results with even 3 iterations) after which all the random walks are stopped simultaneously.
The \texttt{apply()} function is responsible for keeping track of the number of walkers that have stopped on each vertex and \texttt{scatter()} distributes the walkers still alive to the neighbors of the vertex. 
The \texttt{scatter()} phase is the most challenging part of the implementation. In order to reduce information exchange between machines, we use a couple of ideas.

First, we notice that random walks do not have identity. Hence, random walks destined for the same neighbor can be combined into a single message. The second optimization and significant part of our work is modifying the GraphLab engine. The recent versions of GraphLab (since PowerGraph) partition the graph by splitting vertices. As a consequence, the engine will need to synchronize all the mirrors of a vertex over the network a number of times during each GAS cycle.

When running a few random walks, only a handful of neighbors end up receiving walkers. For this reason, synchronizing all mirrors can be very wasteful. We deal with that by implementing \emph{randomized synchronization}. We expose parameter $p_s \in [0,1]$ to the user as a small extension to the GraphLab API. It describes the fraction of replicas that will be synchronized. Replicas not synchronized remain idle for the upcoming scatter phase. The above  edits in the engine are only a matter of a few (about 10) lines of code. 
Note that the $p_s$ parameter is completely optional, i.e., setting it to 1 will result in the original engine operation. Hence, other analytic workloads will not be affected. However, any random walk or ``gossip'' style algorithm (that sends a single messages to a random subset of its neighbors) can benefit by exploiting $p_s$.
Our modification of the GraphLab engine as well as the \rwa\  vertex program  can be found in \cite{borokhovich2014frogwildcode}.

\subsection{Results}

$\rwa$ is significantly faster and uses less network and CPU compared to $\gla$. Let us start with the Twitter graph and the AWS system. In Figure~\ref{fig:plots221}(a) we see that, while  $\gla$ takes about $7.5$ seconds per iteration (for 12 nodes),  $\rwa$ takes less than $1$ sec, achieving more than a $7x$ speedup. Reducing the value of $p_s$   decreases the running time. We see a similar picture when we study the total running time of the algorithms in Figure \ref{fig:plots221}(b)).

We plot network performance in Figure~\ref{fig:plots221}(c). We get a $1000x$ improvement compared to the exact $\gla$, and more than $10x$ with respect to doing $1$ or $2$ iterations of $\gla$. 
In Figure~\ref{fig:plots221}(d) we can see that the total CPU usage reported by the engine is also much lower for $\rwa$.

We now turn to compare the approximation metrics for the PageRank algorithm. For various $k$, we check the two accuracy metrics: Mass captured (Figure \ref{fig:plots222}(a)) and the Exact identification (Figure \ref{fig:plots222}(b)).  
Mass captured -- is the total PageRank that the reported top-$k$ vertices worth in the exact ranking. Exact identification -- is the number of vertices in the intersection of the reported top-$k$ and the exact top-$k$ lists. We can see that the approximation achieved by the $\rwa$ for $p_s=1$ and $p_s=0.7$ always outperforms the $\gla$ with 1 iteration. The approximation achieved by the $\rwa$ with $p_s=0.4$ is relatively good for the both metrics, and with $p_s=0.1$ is reasonable for the \emph{Mass captured} metrics. 

\begin{figure*}[!t]
\centering
\includegraphics[width=1\columnwidth]{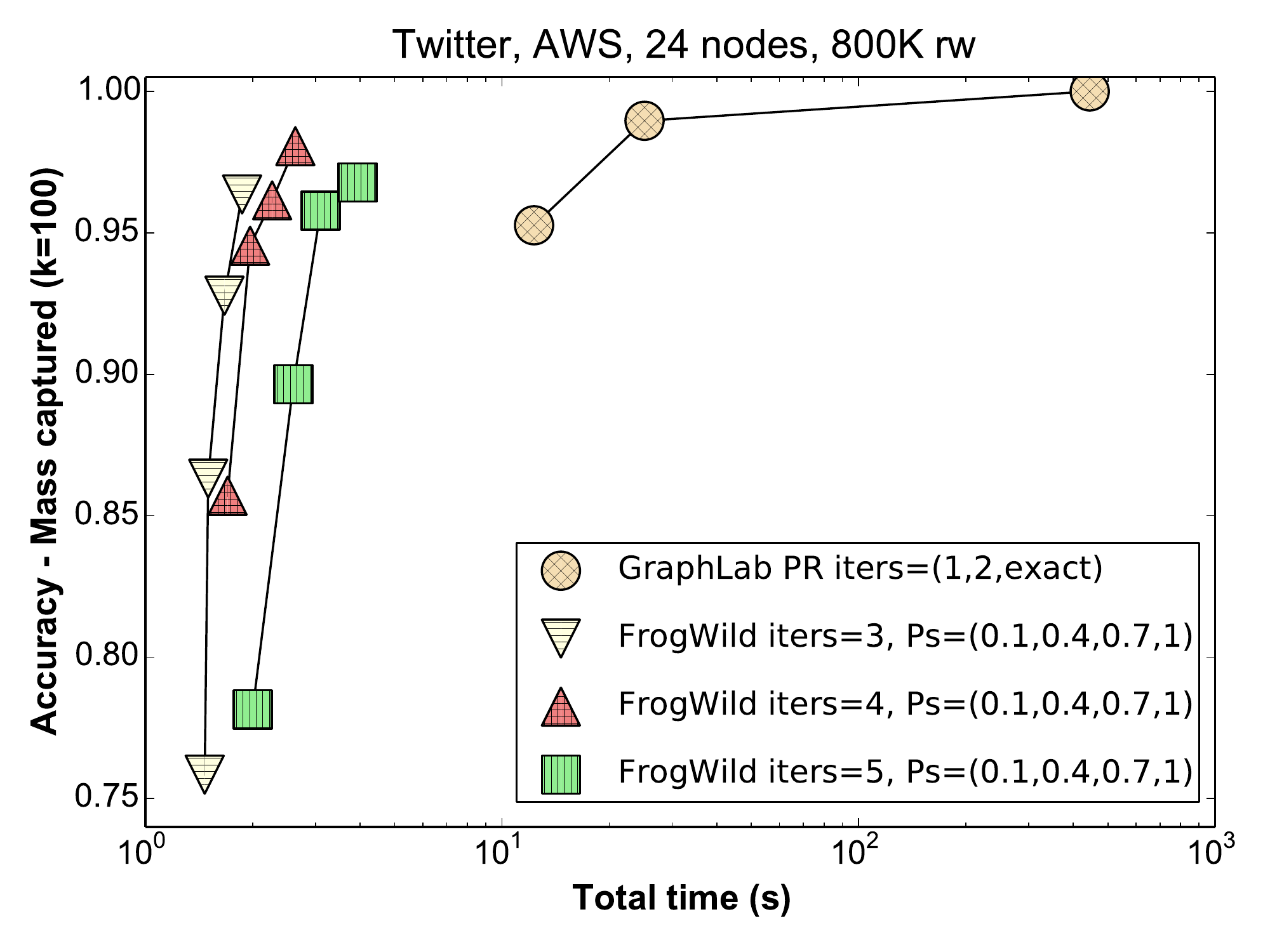}
\includegraphics[width=1\columnwidth]{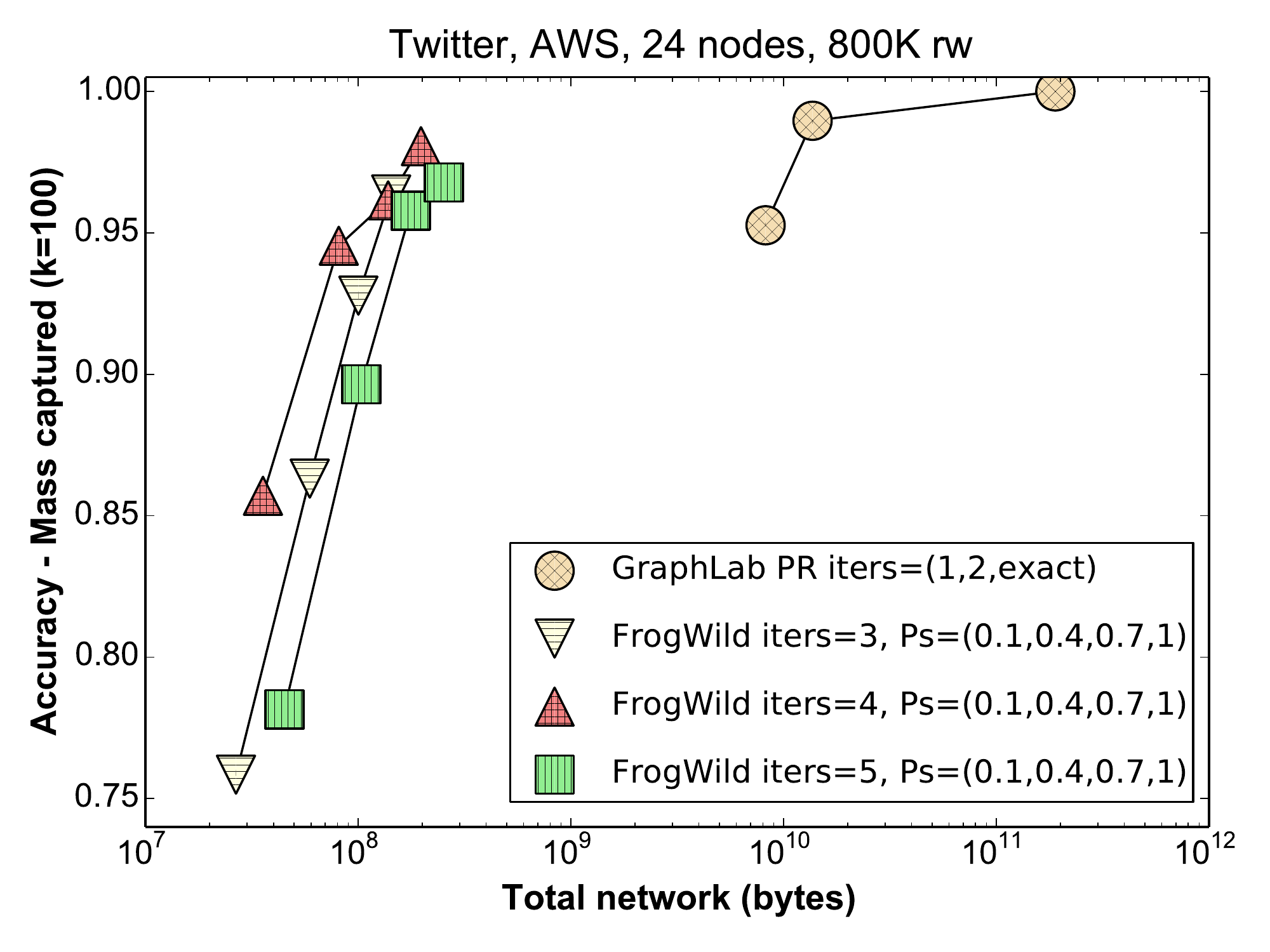}\\
{\small(a)\hspace{\columnwidth}(b)}
\caption{\small 
PageRank approximation accuracy with the ``Mass captured'' metric for top-100 vertices. Graph: Twitter; system: AWS (Amazon Web Services) with 24 nodes; $\rwa$ parameters: 800K initial random walks.
(a) - Accuracy versus total running time. 
(b) - Accuracy versus total network bytes sent.
}\label{fig:plots223}
\end{figure*}

In Figure \ref{fig:plots223} we can see the tradeoff between the accuracy, total running time, and the network usage. The performance of $\rwa$ is evaluated for various number of iterations and the values of $p_s$. The results show that with the accuracy comparable to  $\gla$,  $\rwa$ has much less running time and network usage. 
Figure \ref{fig:plots224} illustrates how much network traffic we save using $\rwa$. The area of each circle is proportional to the number of bytes sent by each algorithm.

\begin{figure}[!h]
\centering
\includegraphics[width=1\columnwidth]{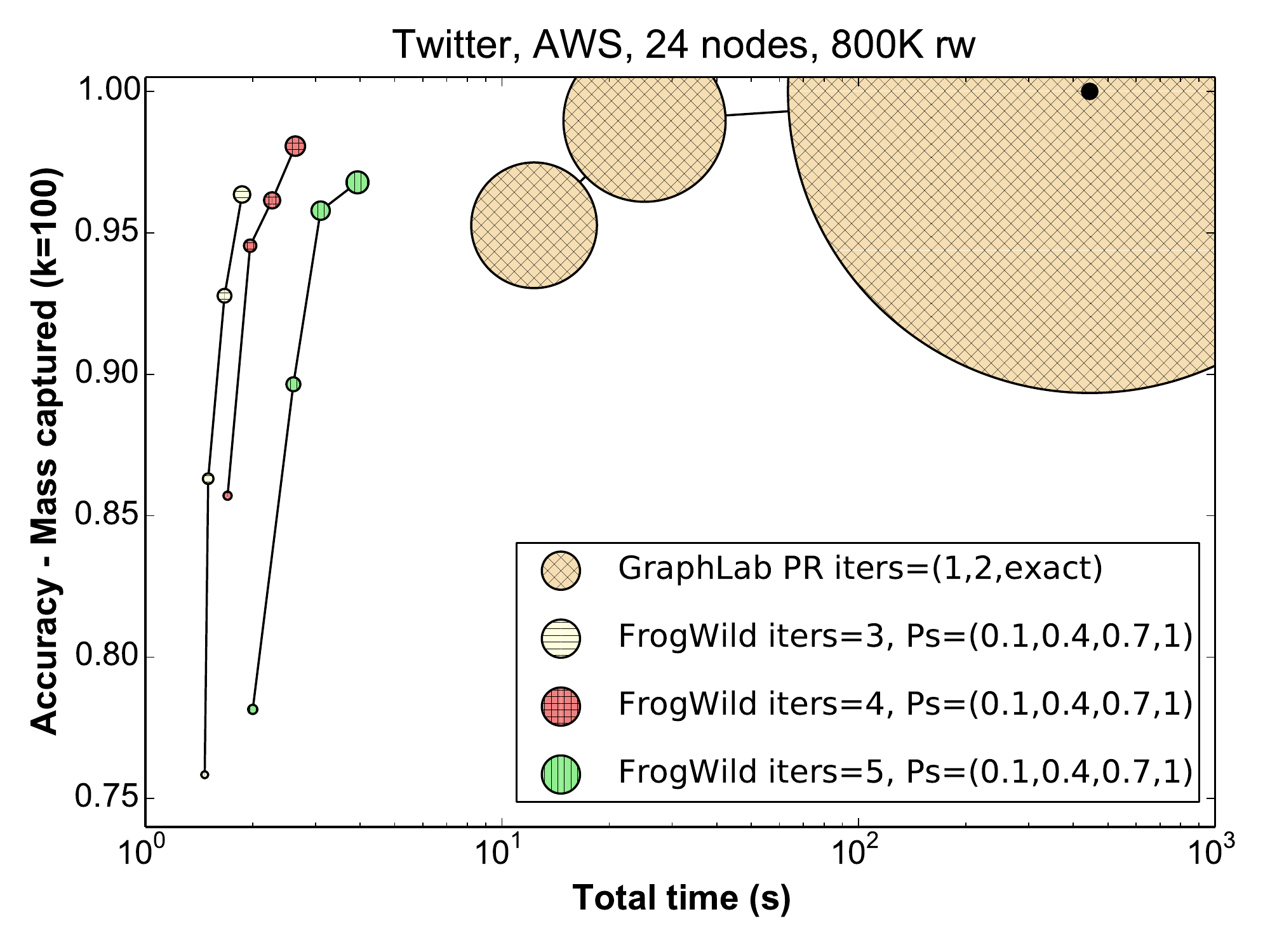}
\caption{\small 
Accuracy versus total running time. Graph: Twitter; system: AWS (Amazon Web Services) with 24 nodes; $\rwa$ parameters: 800K initial random walks. The area of each circle is proportional to the total network bytes sent by the specific algorithm. 
}\label{fig:plots224}
\end{figure}

\begin{figure}[!t]
\centering
\includegraphics[width=1\columnwidth]{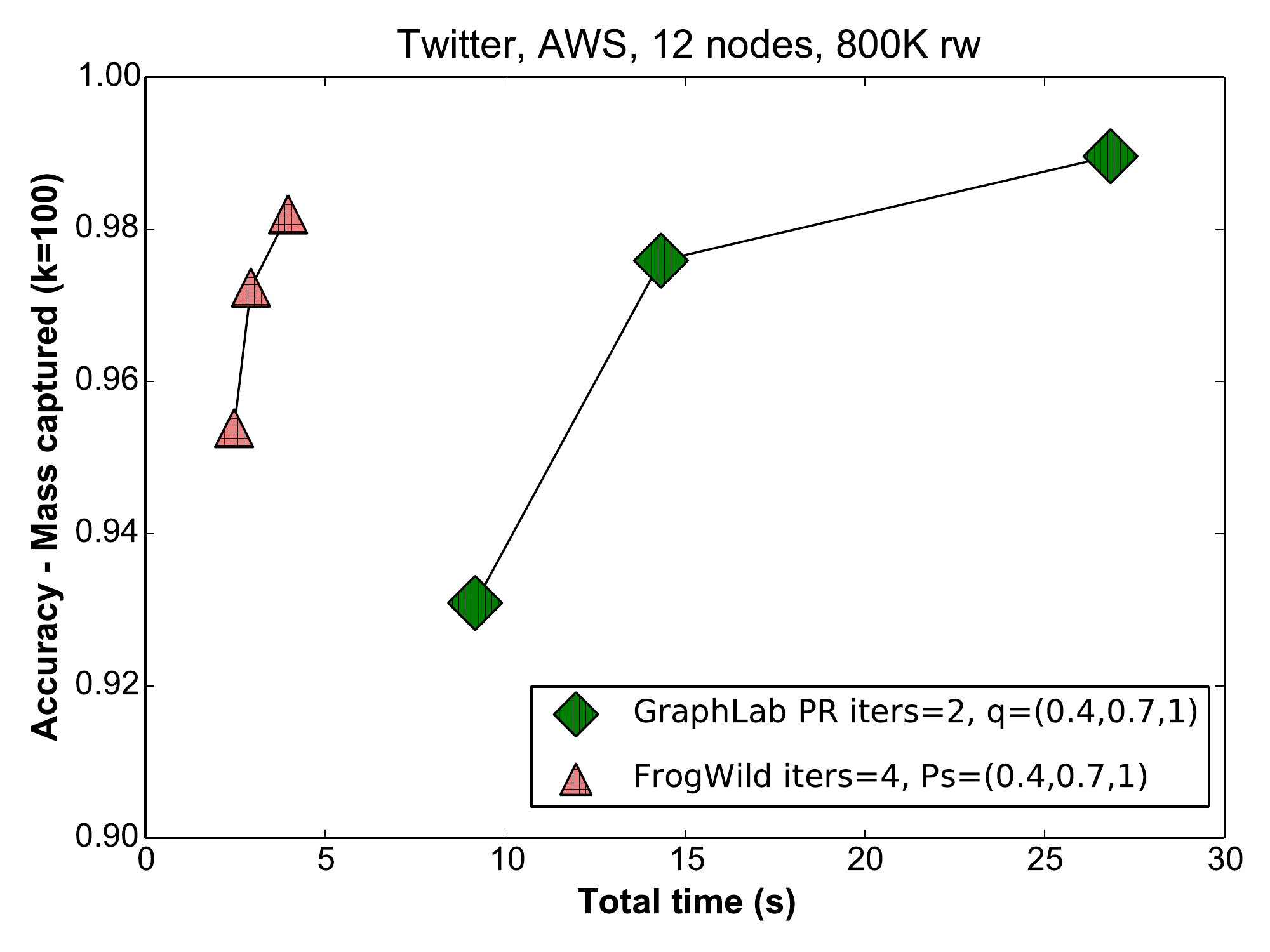}
\caption{\small 
Accuracy versus total running time. Graph: Twitter; system: AWS (Amazon Web Services) with 12 nodes; $\rwa$ parameters: 800K initial random walks. $q=1-r$ is the probability of \emph{keeping} an edge in the sparsification process.
}\label{fig:plots530}
\end{figure}

We also compare $\rwa$ to an approximation strategy that uses a simple sparsification technique described in Section~\ref{sec:prior_work}. First, the graph is sparsified by deleting each edge with probability $r$, then $\gla$ is executed.
In Figure \ref{fig:plots530}, we can see that $\rwa$ outperforms this approach in terms running time while achieving comparable accuracy.

\begin{figure*}[!hp]
\centering
\includegraphics[width=1\columnwidth]{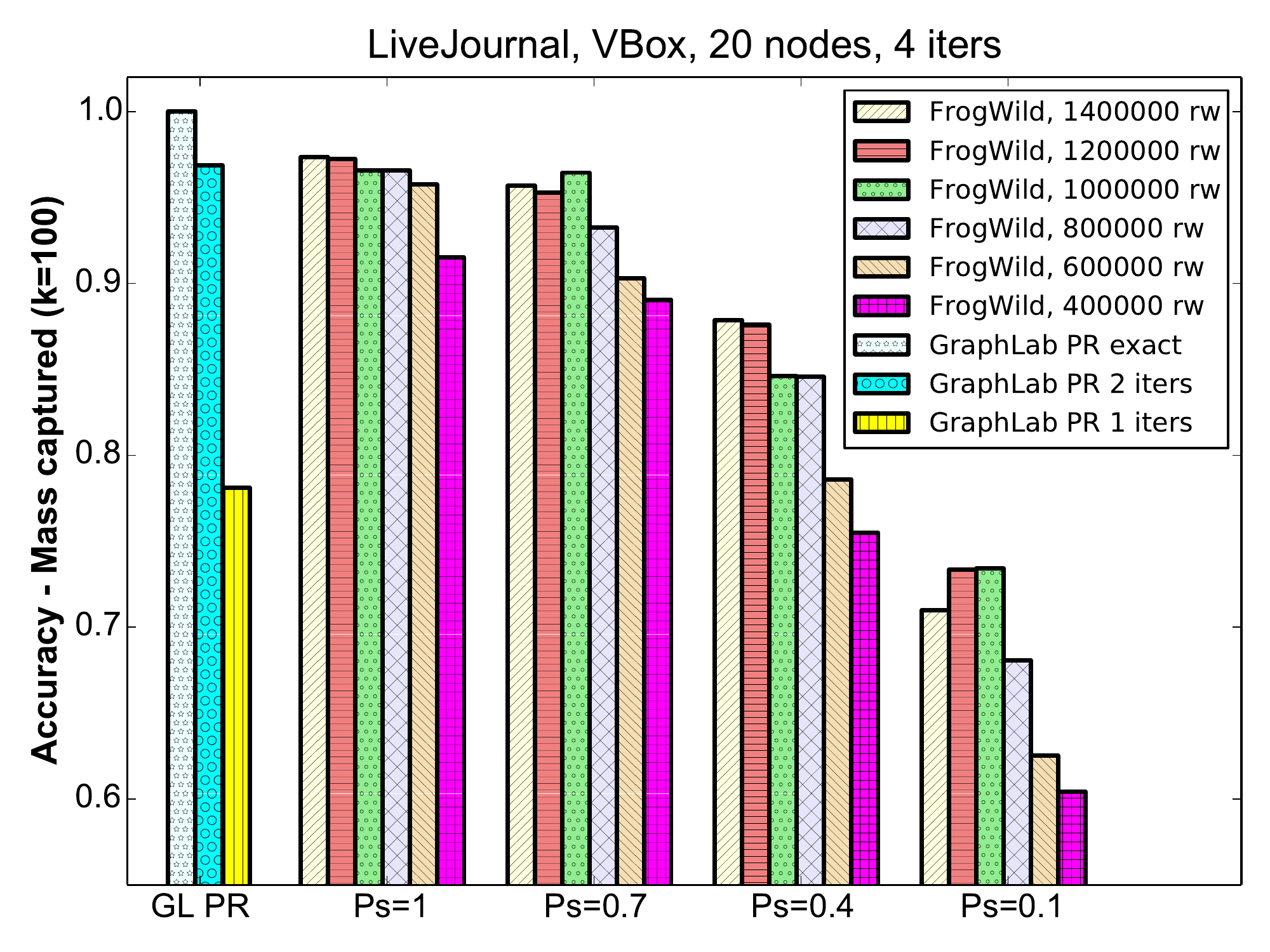}
\includegraphics[width=1\columnwidth]{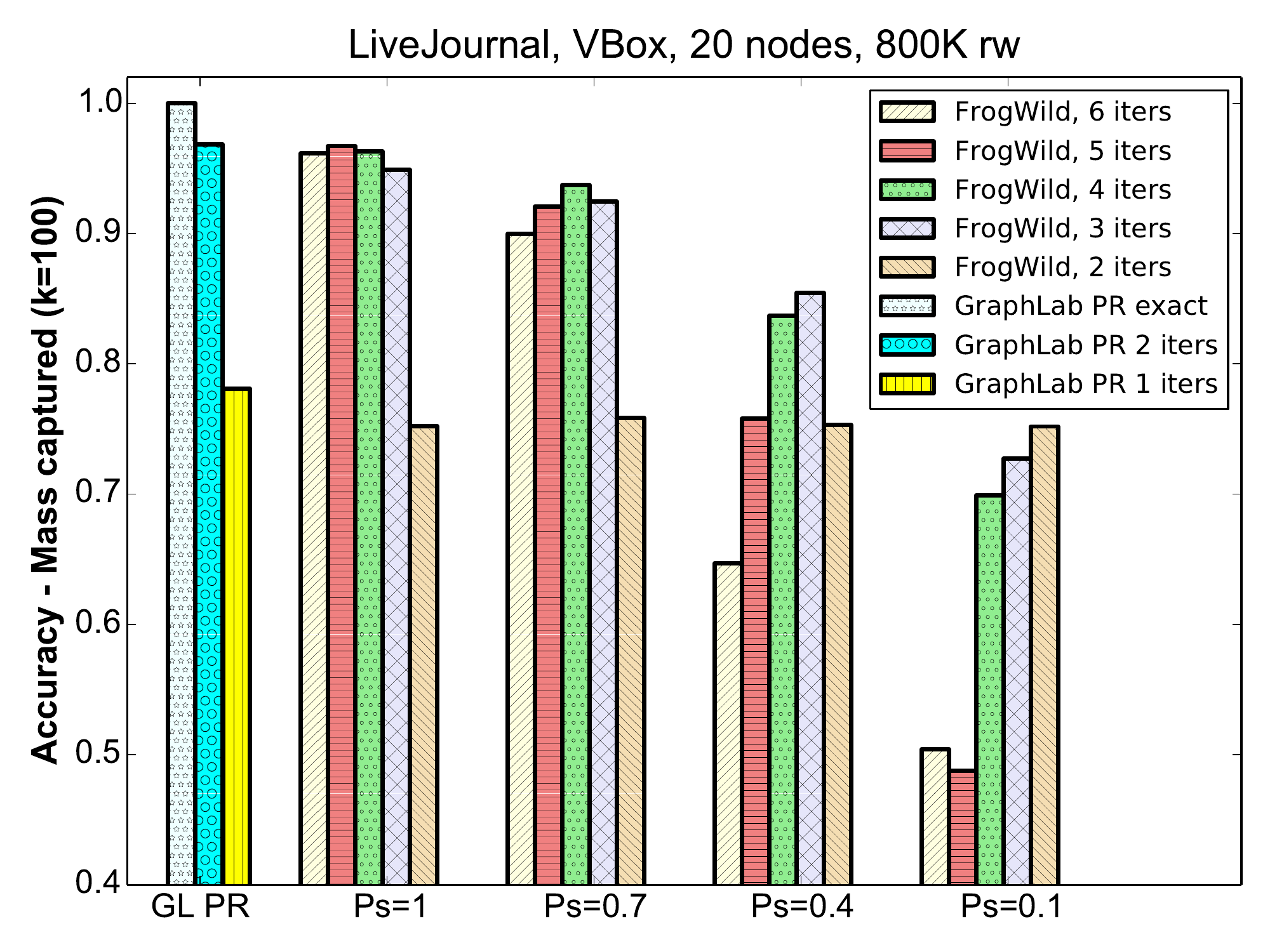}\\
{\small(a)\hspace{\columnwidth}(b)}\\\vspace{2mm}
\includegraphics[width=1\columnwidth]{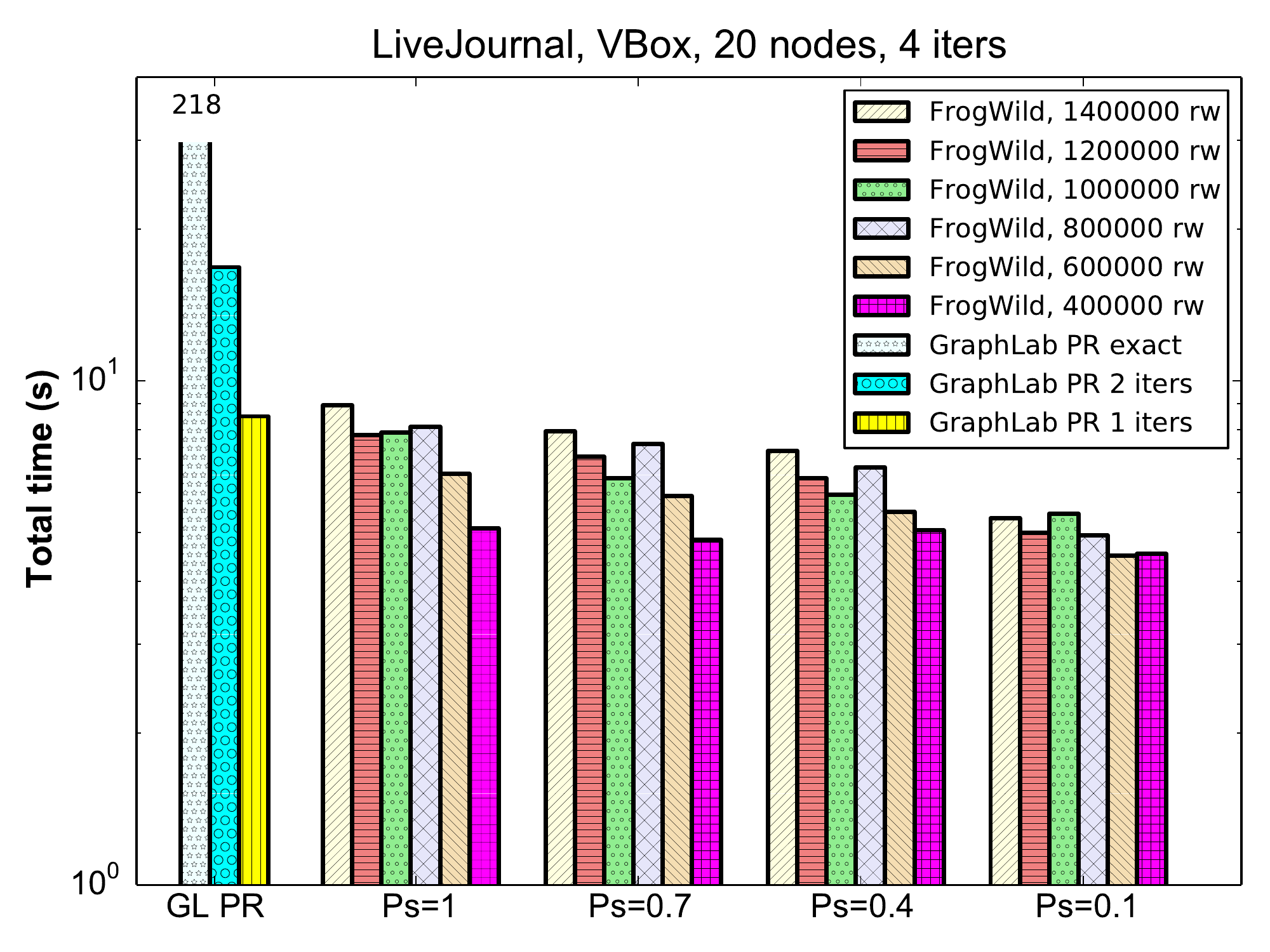}
\includegraphics[width=1\columnwidth]{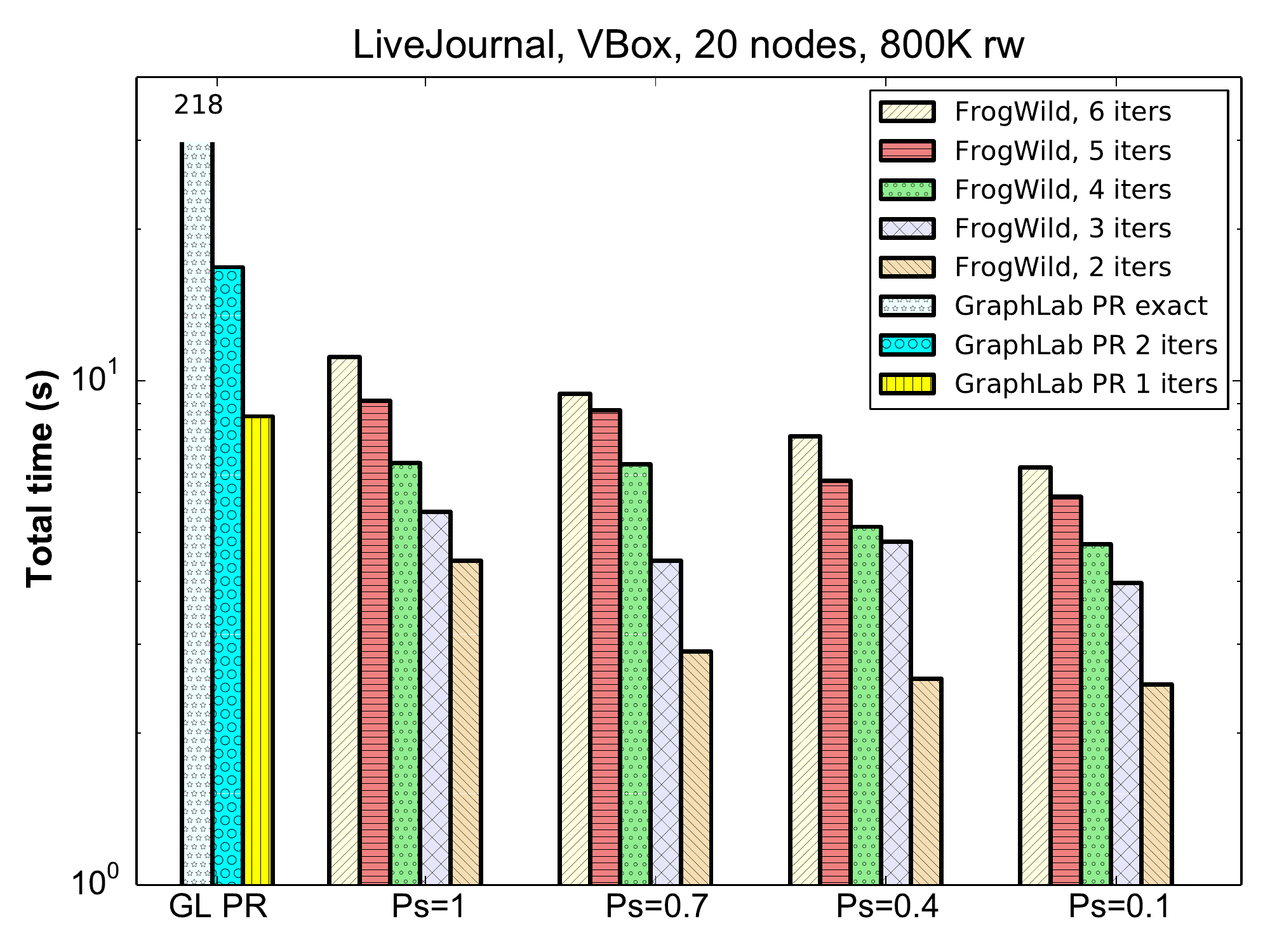}\\
{\small(c)\hspace{\columnwidth}(d)}
\vspace{-2mm}
\caption{\small 
Graph: LiveJournal; system: VirtualBox with 20 nodes.
(a) -- Accuracy for various number of initial random walks in the $\rwa$ (with 4 iterations).
(b) -- Accuracy for various number of iterations of $\rwa$ (with 800K initial random walks).
(c) -- Total running time for various number of initial random walks in the $\rwa$ (with 4 iterations).
(d) -- Total running time for various number of iterations of $\rwa$ (with 800K initial random walks).
}\label{fig:plots225}
\end{figure*}

Finally, we plot results for the LiveJournal graph on the VirtualBox system. Figures~\ref{fig:plots225}(a,b) show the effect of the number of walkers, $N$, and the number of iterations for $\rwa$ on the achieved accuracy. Good accuracy and running time (see Figure \ref{fig:plots225}(c,d)) are achieved for $800K$ initial random walks and $4$ iterations of $\rwa$. Similar to the Twitter graph, also for the LiveJournal graph we can see, in Figure~\ref{fig:plots227}, that our algorithm is faster and uses much less network, while still maintaining good PageRank accuracy. By varying the number of initial random walks and the number of iterations we can fine-tune the $\rwa$ for the optimal accuracy-speed tradeoff.

Interestingly, for both graphs (Twitter and LiveJournal), reasonable parameters are: 800K initial random walks and 4 iterations, despite the order of magnitude difference in the graph sizes. This implies slow growth for the necessary number of frogs with respect to the size of the graph. This scaling behavior is tough to check in practice, but it is explained by our analysis. Specifically, Remark~\ref{thm:scaling} shows that the number of frogs should scale as
$N = O\left({k \over \mu_k(\pi)^2}\right)$. 

\begin{figure*}[!bp]
\centering
\includegraphics[width=1\columnwidth]{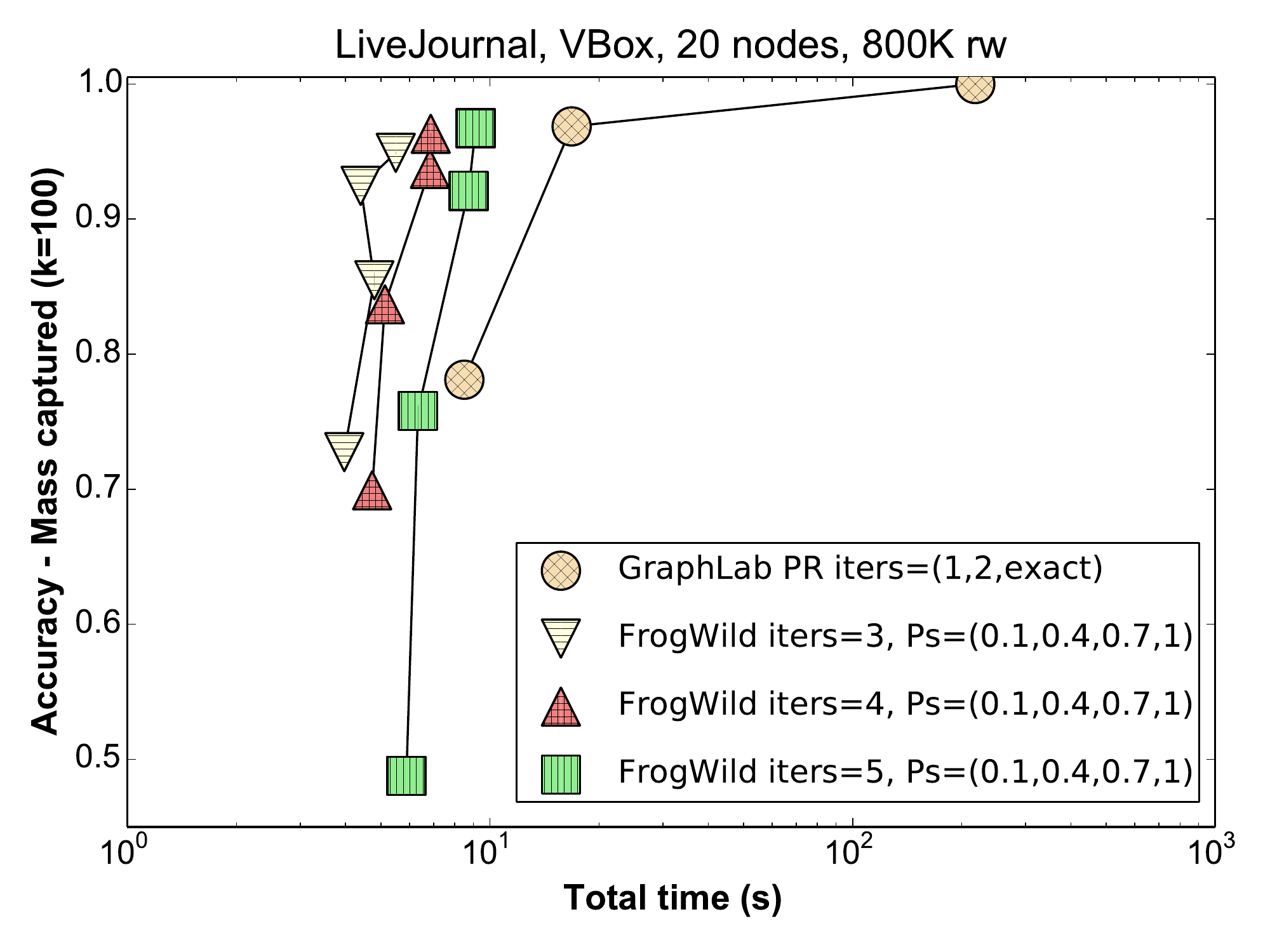}
\includegraphics[width=1\columnwidth]{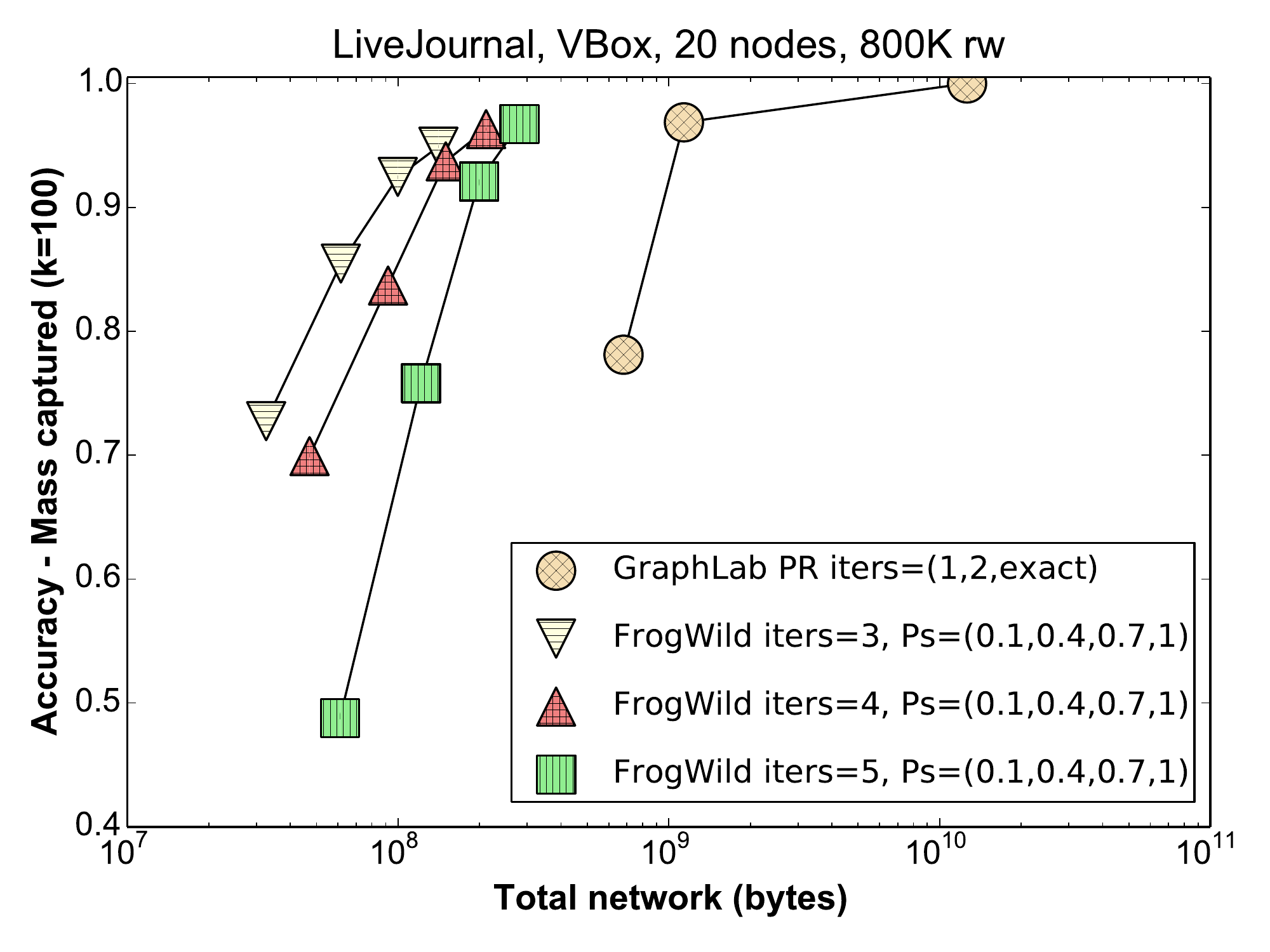}\\
{\small(a)\hspace{\columnwidth}(b)}
\vspace{-2mm}
\caption{\small 
Graph: LiveJournal; system: VirtualBox with 20 nodes; $\rwa$ parameters: 800K initial random walks. 
(a) -- Accuracy versus total running time.
(b) -- Accuracy versus total network bytes sent.
}\label{fig:plots227}
\end{figure*}

\begin{figure}[!t]
\centering
\includegraphics[width=1\columnwidth]{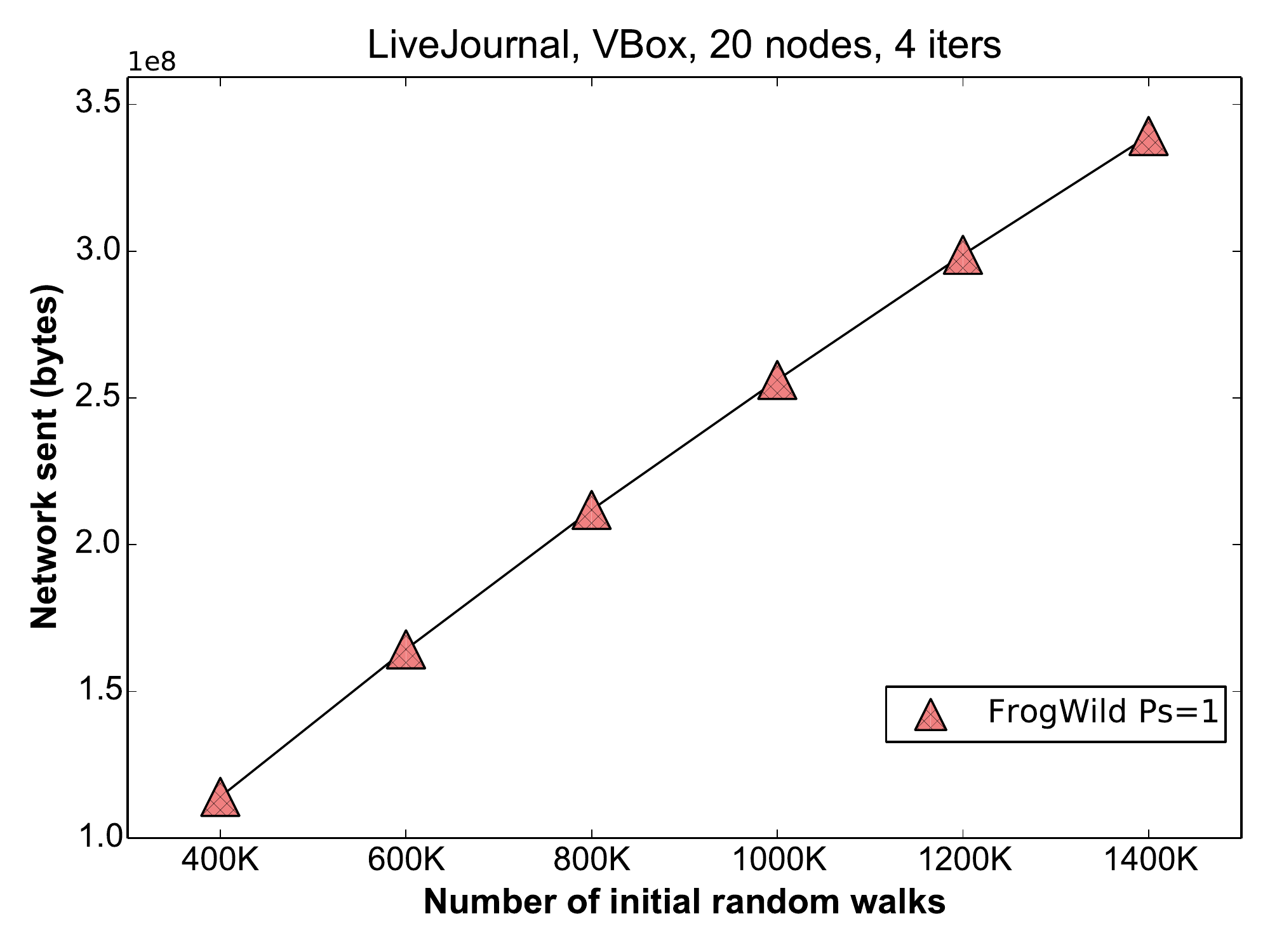}
\caption{\small 
Network usage of $\rwa$ versus the number of initial random walks. 
Graph: LiveJournal; system: VirtualBox with 20 nodes; $\rwa$ parameters: 4 iterations.
}\label{fig:plots532}
\end{figure}

\bibliographystyle{abbrv}

\begin{appendix}

\section{Edge Erasure Model}
\label{sec:erasuremodel}

\begin{definition}[Edge Erasure Model]
\label{def:erasuremodel}
An edge erasure model is a process that is independent from the random walks (up to time $t$) and  temporarily erases a subset of all edges at time $t$. The event $E_{i,j}^t$ represents the erasure of edge $(i,j)$ from the graph for time $t$. The edge is not permanently removed from the graph, it is just disabled and considered again in the next step. The edge erasure models we study satisfy the following properties.
\begin{enumerate}
\item Edges are erased independently for different vertices,
\[	
	\mathbb{P}(E_{i,j}^t, E_{i,k}^t )	 = \mathbb{P}(E_{i,j}^t )\mathbb{P}(E_{i,k}^t )
\]
and across time,
\[	
	\mathbb{P}(E_{i,j}^t, E_{i,j}^s )	 = \mathbb{P}(E_{i,j}^t )\mathbb{P}(E_{i,j}^s ).
\]
\item Each outgoing edge is preserved (not erased) with probability at least $p_s$.
\[	
	\mathbb{P}(\overline{E_{i,j}^t}) \geq p_s
\]
\item Erasures do not exhibit significant negative correlation. Specifically,
\[	
	\mathbb{P}(\overline{E_{i,j}^t}|\overline{E_{i,k}^t}) \geq p_s.
\]
\item Erasures in a neighbourhood are symmetric. Any subset of out-going edges of vertex $i$, will be erased with exactly the same probability as another subset of the same cardinality.
\end{enumerate}
\end{definition}

The main two edge erasure models we consider are described here. They both satisfy all required properties. Our theory holds for both\footnote{ {\em Independent Erasures} can lose some walkers, when it temporarily leads to some nodes having zero out-degree.}, but in our implementation and experiments we use "At Least One Out-Edge Per Node."

\begin{example}[Independent Erasures]
Every edge is preserved independently with probability $p_s$.
\end{example}

\begin{example}[At Least One Out-Edge Per Node]
This edge erasure model, decides all erasures for node $i$ independently, like {\em Independent Erasures}, but if all out-going edges for node $i$ are erased, it draws and enables one of them uniformly at random.
\end{example}

\section{Theorem Proofs}

\subsection{Proof of Theorem~\ref{thm:main}}
\label{sec:proofmain}

\ifx\omitproofs\undefined
In this section we provide a complete proof of our main results. We start from simple processes and slowly introduce the analytical intricacies of our system one-by-one giving guarantees on the performance of each stage.
\else
In this section we provide a proof sketch for our main results. We start from simple processes and slowly introduce the analytical intricacies of our system one-by-one giving guarantees on the performance of each stage. Detailed proofs for the lemmata are deferred to the long version of this paper \cite{mitliagkas2014frogwildlong}.
\fi

\begin{process}[Fixed Step]
\label{proc:fixedstep}
Independent walkers start on nodes selected uniformly at random and perform random walks on the {\em augmented} graph. This means that teleportation happens with probability $p_T$ and the walk is described by the transition probability matrix (TPM) $Q$, as defined in Section~\ref{sec:problemformulation}. Each walker performs exactly $t$ steps before yielding a sample. The number of walkers tends to infinity.
\end{process}

Before we talk about the convergence properties of this Markov chain, we need some definitions.

\begin{definition}
\label{def:chicontrast}
The $\chi^2$-contrast $\chi^2(\alpha;\beta)$ of $\alpha$ with respect to $\beta$ is defined by
\[
	\chi^2(\alpha;\beta) = \sum_i { (\alpha(i) - \beta(i) )^2 \over \beta(i)  }.
\]
\end{definition}

\begin{lemma}
\label{thm:chiub}
Let $\pi \in \Delta^{n-1}$ a distribution satisfying \\ 
\[
\min_i \pi(i) \geq {c \over n}
\]
for constant $c\leq 1$, and let $u \in \Delta^{n-1}$ denote the uniform distribution. Then, $
	\chi^2(u;\pi) \leq \left( { 1 - c \over c } \right)$.
\end{lemma}

\ifx\omitproofs\undefined
\begin{proof}
\begin{align*}
	\chi^2(u;\pi) 
	&= \sum_i { (1/n - \pi(i) )^2 \over \pi(i)  } 
	= \sum_i \left( { 1 \over n^2 \pi(i) } - {1 \over n}  \right) \\
	&= {1 \over n } \sum_i  { 1 - n \pi(i) \over n \pi(i) }  
	\leq {1 \over n } \sum_i  { 1 -c \over c }  
	= { 1 -c \over c } 
\end{align*}
Here we used the assumed lower bound on $\pi(i)$ and the fact that $(1-x)/x$ is decreasing in $x$.
\end{proof}
\fi

\begin{lemma}
\label{thm:fixedstepchi}
Let $\pi^t$ denote the distribution of the walkers after $t$ steps. Its $\chi^2$-divergence with respect to the PageRank vector, $\pi$, is
\[
	\chi^2(\pi^t;\pi) \leq \left({1-p_T \over p_T}\right) (1-p_T)^t.
\]
\end{lemma}

\ifx\omitproofs\undefined
\begin{proof}
	Since $Q$ is ergodic (but non-reversible) we can use the contrast bound in \cite{pierre1999markov}, which gives us
\[
	\chi^2(\pi^t;\pi) \leq \lambda_2^t(\tilde{Q}Q) \chi^2(\pi_0;\pi),
\]
where $\tilde{Q} = DQ'D^{-1}$, for $D=diag(\pi)$, is called the {\em multiplicative reversibilization} of $Q$. We want an upper bound on the second largest eigenvalue of $\tilde{Q}Q	 = D Q ' D^{-1} Q $. From the Perron-Frobenius theorem, we know that $\lambda_1(Q)=1$ and from \cite{haveliwala2003second,elden2004note,serra2005jordan},
$|\lambda_2(Q)| < 1 - p_T$.  Matrix $Q$ is similar to
$\tilde{Q}=D Q ' D^{-1}$, so they have the same spectrum. From this we get the  bound
\[
	|\lambda_2(\tilde{Q}Q)|  \leq 1 - p_T.
\]
The starting distribution $\pi_0$ is assumed to be uniform and every element of the PageRank vector is lower bounded by $p_T/n$. From Lemma~\ref{thm:chiub}, we get
\[
	\chi^2(\pi_0;\pi) \leq \left( 1 - p_T \over p_T \right),
\]
and putting everything together we get the statement.
\end{proof}
\fi 

\begin{process}[Truncated Geometric]
\label{proc:truncgeom}
Independent \\
walkers start on nodes selected uniformly at random and perform random walks on the {\em original} graph. This means that there is no teleportation and the walk is described by the TPM $P$ as defined in Section~\ref{sec:problemformulation}. Each walker performs a random number of steps before yielding a sample. Specifically, the number of steps follows a geometric distribution with parameter $p_T$. Any walkers still active after $t$ steps are stopped and their positions are acquired as samples. This means that the number of steps is given by the minimum of $t$ and a geometric random variable with parameter $p_T$. The number of walkers tends to infinity.
\end{process}

\begin{lemma}
\label{thm:truncatedequivalence}
The samples acquired from Process~\ref{proc:fixedstep} and Process~\ref{proc:truncgeom} follow the exact same distribution.

\end{lemma}

\ifx\omitproofs\undefined
\begin{proof}
Let $\pi_t$ denote the distribution of the walk after $t$ steps according to $Q$ (Process~\ref{proc:fixedstep}) and $\pi_t'$ denote the distribution of the samples provided by the truncated geometric process (Process~\ref{proc:truncgeom}). Note that the both have the same uniform starting distribution $\pi_0=\pi_0'= u = 1_{n\times 1}/n$. For the latter process, the sampling distribution is
\begin{equation}
\label{eqn:distrgeomtrunc}
	\pi_t' = \sum_{\tau=0}^t p_T (1- p_T)^\tau P^\tau u
	+ (1-p_T)^{t+1}P^t u .
\end{equation}
The last term corresponds to the cut-off we impose at time $t$. Now consider the definition of the TPM $Q$ (Definition~\ref{def:pagerank}). The Markov chain described by $Q$, teleports at each steps with probability $p_T$; otherwise, it just proceeds according to the TPM $P$. With every teleportation, the walker starts from the uniform distribution, $u$ -- any progress made so far is completely "forgotten." Therefore, we just need to focus on the epoch between the last teleportation and the cut-off time $t$.
The times between teleportation events are geometrically distributed with parameter $p_T$. The teleportation process is memory-less and reversible. Starting from time $t$ and looking backwards in time, the last teleportation event is a geometric number of steps away, and with probability $(1-p_T)^{t+1}$ it happens before the starting time $0$. In that case we know that no teleportation happens in $[0, t]$. The samples acquired from this process are given by
\begin{equation}
	\pi_t = \sum_{\tau=0}^t p_T (1- p_T)^\tau P^\tau u
	+ (1-p_T)^{t+1}P^t u,
\end{equation}
which is exactly the distribution for Process~\ref{proc:truncgeom} given in \eqref{eqn:distrgeomtrunc}.
\end{proof}
\fi

\begin{lemma}[Mixing Loss]
\label{thm:mixingloss}
Let $\pi^t \in \simplexn$ denote the distribution of the samples acquired through Process~\ref{proc:truncgeom}. The mass it captures (Definition~\ref{def:capturedmass}) is lower-bounded as follows.
\begin{equation*}
   \mu_k(\pi^t) \geq \mu_k(\pi) - \sqrt{(1-p_T)^{t+1} \over p_T }
\end{equation*}
\end{lemma}

\ifx\omitproofs\undefined
\begin{proof}
Let us define $\delta_i = \pi^t_i - \pi_i $. First we show that
\begin{equation}
\label{eqn:massl1}
\mu_k(\pi^t) \geq \mu_k(\pi) - \|\pi - \pi^t\|_1.
\end{equation}
To see this, first consider the case when $\delta_1=-\delta_2$ and $\delta_i=0$ for $i=3,...,n$. The maximum amount of mass that can be missed by $\pi^t$, in this case, is $|\delta_1|+|\delta_2|$. This happens when $\pi_1$ and $\pi_2$ are exactly $|\delta_1|+|\delta_2|$ apart and are flipped in the ordering by $\pi^t$. This argument generalizes to give us \eqref{eqn:massl1}.
Now assume that the $\chi^2$-divergence of $\pi^t$ with respect to the PageRank vector $\pi$ is bounded by $\epsilon^2$.
Now using a variational argument and the KKT conditions we can show that setting $\delta_i = \pi_i \epsilon$ for all $i$ gives the maximum possible $l_1$ error:
\begin{equation}
\label{eqn:l1chi2}
|| \pi - \pi^t ||_1 \leq \epsilon = \sqrt{ \chi^2(\pi^t;\pi) }.
\end{equation}
For another proof using the Cauchy-Schwarz inequality, see \cite{pierre1999markov}.
Finally, combining \eqref{eqn:l1chi2} with \eqref{eqn:massl1} and the results from Lemma~\ref{thm:truncatedequivalence} and \ref{thm:fixedstepchi} gives us the statement.
\end{proof}
\fi

\begin{lemma}[Sampling Loss]
\label{thm:samplingloss}
Let $\hat{\pi}_N$ be the estimator of Definition~\ref{def:estimator} using $N$ samples from the \rwa!\ system. This is essentially, Process~\ref{proc:truncgeom} with the added complication of {\em random synchronization} as explained in Section~\ref{sec:algo}. Also, let $\pi^t$ denote the sample distribution after $t$ steps, as defined in Lemma~\ref{thm:truncatedequivalence}. The mass captured by this process is lower bounded as follows, with probability at least $1-\delta$.
\[
\mu_k(\hat{\pi}_N)  \geq \mu_k(\pi^t) - 	 \sqrt{
		{ k \over \delta}
		\left[
			{ 1  \over N } + (1 - p_s^2)\pmeet(t)
		\right]
	},
\]
\end{lemma}

\ifx\omitproofs\undefined
\begin{proof}
In this proof, let $x_l^t$ denote the individual (marginal) walk distribution for walker $l$ at time $t$.
We know, that it follows the dynamics $x_l^{t+1} = P x_l^{t}$, for all $l \in [N]$, i.e.\  $x_l^t = x_1^t$. First we show that $\| \hat{\pi}_N - x_1^t\|_2$ is small. 
\begin{align}
	\label{eqn:cheby}
	\mathbb{P}( \| \hat{\pi}_N - x_1^t\|_2 > \epsilon )
	&\leq {\mathbb{E}[ \| \hat{\pi}_N - x_1^t\|_2^2] \over  \epsilon^2 } 
\end{align}
Here we used Markov's inequality. We use $s_l^t$ to denote the position of walker $l$ at time $t$ as a vector. For example, $s_l^t = e_i$, if walker $l$ is at state $i$ at time $t$.
Now let us break down the norm on the numerator of \eqref{eqn:cheby}.
\begin{align}
	\label{eqn:empiricalvariance}
	\| \hat{\pi}_N &- x_1^t\|_2^2
		= \bigg\Vert {1 \over N} \sum_l ( s_l^t - x_1^t)\bigg\Vert_2^2 \notag \\
	    &= {1 \over N^2} \sum_l \| s_l^t - x_1^t\|_2^2
	    	+ {1 \over N^2} \sum_{l \neq k} ( s_l^t - x_1^t)'( s_k^t - x_1^t)
\end{align}
For the diagonal terms we have:
\begin{align}
	\label{eqn:samplevariance}
	\mathbb{E}[\| s_l^t - x_1^t\|_2^2]
	&=\sum_{i \in [n]} \mathbb{E}\left[\| s_l^t - x_1^t\|_2^2 | s_l^t=i \right] \mathbb{P}(s_l^t=i) \notag \\
		&=\sum_{i \in [n]} \| e_i^t - x_1^t\|_2^2 x_1^t(i) 
		= 1 - \| x_1^t \|_2^2 \leq 1
\end{align}

Under the edge erasures model, the trajectories of different walkers are not generally independent. For example, if they happen to meet, they are likely to make the same decision for their next step, since they are faced with the same edge erasures. Now we prove that even when they meet, we can consider them to be independent with some probability that depends on $p_s$. 

Consider the position processes for two walkers, $\{ s_1^t\}_{t}$ and
$\{ s_2^t\}_{t}$. At each step $t$ and node $i$ a number of out-going edges are erased. Any walkers on $i$, will choose uniformly at random from the remaining edges. Now consider this {\em alternative process}.

\begin{process}[Blocking Walk]
\label{def:blockingprocess}
\label{proc:blockingwalk}

A blocking walk on the graph under the erasure model, follows these steps.
\begin{enumerate}
\item Walker $l$ finds herself on node $i$ at time $t$.
\item Walker $l$ draws her next state uniformly from the {\em full} set of out-going edges. 
\[
	w \sim \mathrm{Uniform}(\mathcal{N}_o(i))
\]
\item If the edge $(i,w)$ is erased at time $t$, the walker cannot traverse it. We call this event a {\em block} and denote it by $B_l^t$. In the event of a block:
\begin{itemize}
\item Walker redraws her next step from the out-going edges of $i$ {\em not erased} at time $t$.
\item Otherwise, $w$ is used as the next state.
\end{itemize}
\end{enumerate}
\end{process}

A blocking walk is exactly equivalent to our original process; walkers end up picking a destination uniformly at random among the edges not erased. From now on we focus on this description of our original process. We use the same notation: $\{ s_l^t\}_{t}$ for the position process and $\{ x_l^t\}_{t}$ for the distribution at time $t$.

Let us focus on just two walkers, $\{ s_1^t\}_{t}$ and $\{ s_2^t\}_{t}$ and consider a third process: two independent random walks on the same graph. We assume that these walks operate on the full graph, i.e.\ no edges are erased. We denote their positions by $\{ v_1^t\}_{t}$ and $\{ v_2^t\}_{t}$ and their marginal distributions by $\{ z_1^t\}_{t}$ and $\{ z_2^t\}_{t}$.

\begin{definition}[Time of First Interference]
For two blocking walks, $\tau_I$ denotes the earliest time at which they meet and at least one of them experiences blocking. 
\[
	\tau_{I} = \min \left\{ 
			t : \{s_1^t = s_2^t\} \cap ( B_1^t \cup B_2^t)
	\right\}
\]
We call this quantity the {\em time of first interference}.
\end{definition}

\begin{lemma}[Process equivalence]
\label{thm:equivalence}
For two walkers, the blocking walk and the independent walk are identical until the time of first interference. That is, assuming the same starting distributions, $x_1^0 = z_1^0$ and $x_2^0 = z_2^0$, then
\[
 	x_1^t = z_1^t \quad  \textrm{and} \quad x_2^t = z_2^t \qquad \forall t \leq \tau_I.
\]
\end{lemma}

\begin{proof}
The two processes are equivalent for as long as the blocking walkers make {\em independent} decisions effectively picking {\em uniformly from the full set of edges} (before erasures). From the independence in erasures across time and vertices in Definition~\ref{def:erasuremodel}, as long as the two walkers do not meet, they are  making an independent choices. Furthermore, since erasures are symmetric, the walkers will be effectively choosing uniformly over the full set of out-going edges.

Now consider any time $t$ that the blocking walkers meet. As long as neither of them blocks, they are by definition taking independent steps uniformly over the set of all outgoing edges, maintaining equivalence to the independent walks process. This concludes the proof.
\end{proof}

\begin{lemma}
Let all walkers start from the uniform distribution.
The probability that the time of first interference comes before time $t$ is upper bounded as follows. $
	\mathbb{P}(\tau_I \leq t ) \leq (1 - p_s^2) \pmeet(t)$
\end{lemma}
\begin{proof}
Let $M_t$ be the event of a meeting at time $t$, $M_t \triangleq \left\{   s_1^t = s_2^t  \right\}$.
In the proof of Theorem~\ref{thm:pmeet}, we establish that $\mathbb{P}(M_t) \leq \rho^t/n$,where $\rho$ is the maximum row sum of the transition matrix $P$.
Now denote the event of an interference at time $t$ as follows. $
	I_t \triangleq M_t \cap ( B_1^t \cup B_2^t)$,
where $B_1^t$ denotes the event of blocking, as described in Definition~\ref{def:blockingprocess}. Now, 
\begin{align*}
	\mathbb{P}(I_t)
	 &= \mathbb{P}( M_t \cap ( B_1^t \cup B_2^t))
	 = \mathbb{P}( B_1^t \cup B_2^t | M_t ) \pmeet(t).
\end{align*}
For the probability of a block given that the walkers meet at time $t$,
\begin{align*}
	\mathbb{P}( B_1^t \cup B_2^t | M_t ) 
	&= 1 - \mathbb{P}( \overline{B_1^t} \cap \overline{B_2^t} | M_t ) \\
	&= 1 - \mathbb{P}( \overline{B_2^t} |\overline{B_1^t}, M_t ) \mathbb{P}( \overline{B_1^t} \ | M_t ) 
	\leq 1 - p_s^2.
\end{align*}
To get the last inequality we used, from Definition~\ref{def:erasuremodel}, the lower bound on the probability that an edge is not erased, and the lack of negative correlations in the erasures.

Combining the above results, we get
\begin{align*}
\mathbb{P}(\tau_I &\leq t) 
= \mathbb{P}\left(\sum_{\tau=1}^t \mathbb{I}_{\{ I_\tau \} } \geq 1\right) 
\leq \mathbb{E}\left[\sum_{\tau=1}^t \mathbb{I}_{\{ I_\tau \} } \right]
= \sum_{\tau=1}^t \mathbb{P}(I_\tau) \\
&\leq \sum_{\tau=1}^t (1-p_s^2) \mathbb{P}(M_\tau) 
= {1 - p_s^2 \over n } \sum_{\tau=1}^t \rho^\tau
= (1 - p_s^2)\pmeet(t)
\end{align*}
which proves the statement. 
\end{proof}

Now we can bound the off-diagonal terms in \eqref{eqn:samplevariance}.
\begin{align*}
	\mathbb{E}\bigg[( s_l^t - x_1^t)'( s_k^t &- x_1^t)\bigg]& \\
	=& \mathbb{E}\bigg[( s_l^t - x_1^t)'( s_k^t - x_1^t) \big| \tau_I \leq t \bigg] 
	\mathbb{P}(\tau_I \leq t) \\
	&+ \mathbb{E}\bigg[( s_l^t - x_1^t)'( s_k^t - x_1^t) \big| \tau_I > t \bigg] 
		\mathbb{P}(\tau_I > t)
\end{align*}
In the second term, the case when $l$, $k$ have not interfered, by Lemma~\ref{thm:equivalence}, the trajectories are independent and the cross-covariance is $0$.
In the first term, the cross-covariance is maximized when $s_l^t = s_k^t$. That is,
\[
\mathbb{E}\bigg[( s_l^t - x_1^t)'( s_k^t - x_1^t) \big| \tau_I \leq t \bigg]
\leq 
\mathbb{E}[\| s_l^t - x_1^t\|_2^2] \leq 1
\]
From this we get
\begin{equation}
 \mathbb{E}\bigg[( s_l^t - x_1^t)'( s_k^t - x_1^t)\bigg]
 \leq (1-p_s^2)\pmeet(t),
\end{equation}
and in combination with \eqref{eqn:samplevariance}, we get from \eqref{eqn:empiricalvariance} that
\begin{align*}
	\mathbb{E}\bigg[ \| \hat{\pi}_N - x_1^t\|_2^2 \bigg]
	    &\leq {1 \over N} 
	    	+ {(N-1)(1-p_s^2)\pmeet(t) \over N} .
\end{align*}
Finally, we can plug this into \eqref{eqn:cheby}, and since all marginals $x_l^t$ are the same, and denoted by $\pi^t$, we get
\begin{align}
\label{eqn:chebyresult}
	\mathbb{P}( \| \hat{\pi}_N - \pi^t\|_2 > \epsilon )
	&\leq {1+(1-p_s^2)\pmeet(t)(N-1) \over N \epsilon^2}.
\end{align}

\newcommand\restr[2]{{
  \left.\kern-\nulldelimiterspace 
  #1 
  \vphantom{\big|} 
  \right|_{#2} 
  }}
 
Let $\restr{\pi^t}{S}$ denote the restriction of the vector $\pi^t$ to the set $S$. That is,
$\restr{\pi^t}{S}(i)=\pi^t(i)$ if $i \in S$ and $0$ otherwise.
Now we show that for any set $S$ of cardinality $k$,
\begin{align}
\label{eqn:epsilonprime}
 | \pi^t(S) - \hat{\pi}_N(S)|
 &\leq \| \restr{(\pi^t - \hat{\pi}_N)}{S} \|_1 \notag 
  \leq \sqrt{k} \| \restr{(\pi^t - \hat{\pi}_N)}{S} \|_2 \notag \\
   &\leq \sqrt{k} \| \pi^t - \hat{\pi}_N \|_2 
\end{align}
Here we used the fact that for $k$-length vector $x$, $\|x\|_1 \leq \sqrt{k} \|x\|_2$
and $\| \restr{x}{S} \| \leq \| x \|$.
We define the top-$k$ sets 
\[
	\hat{S}^* = \mathrm{argmax}_{S \subset [n], |S|=k} \hat{\pi}_N(S)
\]
\[
	S^* = \mathrm{argmax}_{S \subset [n], |S|=k} \pi^t(S).
\]
Per these definitions,
\begin{align}
\label{eqn:maxima}
		\hat{\pi}_N(\hat{S}^*)
		&= \max_{S \subset [n], |S|=k} \hat{\pi}_N(S) \notag \\
		&\geq \hat{\pi}_N(S^*) 
		\geq \pi^t(S^*) - \sqrt{k} \| \pi^t - \hat{\pi}_N \|_2 .
\end{align}
The last inequality is a consequence of \eqref{eqn:epsilonprime}.
Now using the inequality in \eqref{eqn:chebyresult} and denoting the LHS probability as $\delta$, we get the statement of Lemma~\ref{thm:samplingloss}.
\end{proof}
\fi

Combing the results of Lemma~\ref{thm:mixingloss} and Lemma~\ref{thm:samplingloss}, we establish the main result, Theorem~\ref{thm:main}.

\subsection{Proof of Theorem~\ref{thm:pmeet} }
\label{sec:proofpmeet}

\begin{proof}
Let $u \in \simplexn$ denote the uniform distribution over $[n]$, i.e.\ $u_i = 1/n$.
The two walks start from the same initial uniform distribution, $u$, and independently follow the same law, $Q$. Hence, at time $t$ they have the same marginal distribution, $p^{t} = Q^t u$. 
From the definition of the augmented transition probability matrix, $Q$, in Definition~\ref{def:pagerank}, we get that
\[
	\pi_i \geq { p_T \over  n }, \quad \forall i \in [n].
\]
Equivalently, there exists a distribution $q \in \simplexn$ such that
\[
	\pi = p_T u + (1-p_T) q.
\]
Now using this, along with the fact that
$\pi$ is the invariant distribution associated with $Q$ (i.e.\ $\pi = Q^t \pi$ for all $t \geq 0$) we get that for any $t \geq 0$,
\begin{align*}
\| \pi \|_\infty  
&= \| Q^t \pi \|_\infty \\
&= \| Q^t p_T u + Q^t (1-p_T) q \|_\infty \\
&\geq p_T \| Q^t u \|_\infty.
\end{align*}
For the last inequality, we used the fact that $Q$ and $q$ contain non-negative entries. 
Now we have a useful upper bound for the maximal element of the walks' distribution at time t.
\begin{equation}
\label{eqn:distifnty}
	\| p^t \|_\infty = \| Q^t u  \|_\infty \leq  { \| \pi \|_\infty \over  p_T }
\end{equation}
Let $M_t$ be the indicator random variable for the event of a meeting at time $t$.
\[
	M_t = \mathbb{I}_{\left\{  \textrm{walkers meet at time } t  \right\}}
\]
Then, $\mathbb{P}(M_t=1) = \sum_{i=1}^n p^{t}_i p^{t}_i = \| p^{t} \|_2^2$.
Since $p^{0}$ is the uniform distribution, i.e.\ $p^{0}_i = { 1 \over n}$ for all $i$, then $\| p^{0} \|_2^2 = {1 \over n}$.
We can also bound the $l_2$ norm of the distribution at other times.
First, we upper bound the $l_2$ norm by the $l_\infty$ norm.
\[
	\| p \|_2^2 = \sum_i p_i^2 \leq \sum_i p_i \| p \|_\infty = \| p \|_\infty
\]
Here we used the fact that $p_i\geq 0$ and $\sum p_i = 1$.

Now, combining the above results, we get
\begin{align*}
p_{\cap}(t) 
&= \mathbb{P}\left(\sum_{\tau=0}^t M_\tau \geq 1\right) 
\leq \mathbb{E}\left[\sum_{\tau=0}^t M_\tau \right]
= \sum_{\tau=0}^t \mathbb{E}[M_\tau]  \\
&= \sum_{\tau=0}^t \mathbb{P}(M_\tau=1) 
=\sum_{\tau=0}^t \| p^{\tau} \|_2^2 \leq \sum_{\tau=0}^t \| p^{\tau}\|_\infty  \\
&\leq {1 \over n } + {t \| \pi \|_\infty \over p_T}.
\end{align*}
For the last inequality, we used \eqref{eqn:distifnty} for $t \geq 1$ and $\| p^0 \|_2^2=1/n$.
This proves the theorem statement. 
\end{proof}

\subsection{Proof of Proposition~\ref{thm:powerlawmax}}
\label{sec:proofpowerlawmax}

\begin{proof}
The expected maximum value of $n$ independent draws from a power-law distribution with parameter $\theta$, is shown in 
\cite{newman2005power}
to be
\[
	\mathbb{E} x_{max} = O(n^{-{ 1 \over \theta-1}}).
\]
Simple application of Markov's inequality, gives us the statement.
\end{proof}

\newpage

\end{appendix}

% ensure same length columns on last page (might need two sub-sequent latex runs)
\balance

\end{document}